\documentclass[runningheads]{llncs}
\usepackage{caption}
\usepackage{subcaption}
\usepackage{graphicx}
\usepackage{xcolor}
\usepackage{hyperref}

\usepackage{amsmath}
\DeclareMathOperator\supp{supp}
\usepackage{algorithmic}
\usepackage{array}
\usepackage{booktabs,makecell,tabularx}

\newcolumntype{C}[1]{>{\centering\arraybackslash}p{#1}}
\newcolumntype{L}{>{\raggedright\arraybackslash}X}

\newcommand{\comment}[1]{}

\usepackage[utf8x]{inputenc}

\begin{document}
\title{Cyber Deception against Zero-day Attacks: A Game Theoretic Approach\thanks{Research was sponsored by the Army Research Laboratory and was accomplished under Cooperative Agreement Numbers W911NF-19-2-0150 and W911NF-13-2-0045 (ARL Cyber Security CRA). The views and conclusions contained in this document are those of the authors and should not be interpreted as representing the official policies, either expressed or implied, of the Army Research Laboratory or the U.S. Government. The U.S. Government is authorized to reproduce and distribute reprints for Government purposes notwithstanding any copyright notation herein. Branislav Bosansky was also supported by the Czech Science Foundation (no. 19-24384Y).}}
%
%\titlerunning{Abbreviated paper title}
% If the paper title is too long for the running head, you can set
% an abbreviated paper title here
%
\author{Md Abu Sayed\inst{1}\orcidID{0000-0002-5560-9150} \and Ahmed H. Anwar \inst{2}\orcidID{0000-0001-8907-3043}   \and Christopher Kiekintveld\inst{1}\orcidID{0000-0003-0615-9584} \and Branislav Bosansky\inst{3}\orcidID{0000-0002-3841-9515} \and Charles Kamhoua\inst{2}\orcidID{0000-0003-2169-5975}}
%0000-0003-0615-9584
\authorrunning{M. A Sayed et al.}
% First names are abbreviated in the running head.
% If there are more than two authors, 'et al.' is used.
%
\institute{University of Texas at El Paso, TX 79968, USA \\
\email{ msayed@miners.utep.edu, cdkiekintveld@utep.edu} \and
US Army Research Laboratory, MD 20783, USA 
\email{a.h.anwar@knights.ucf.edu,charles.a.kamhoua.civ@mail.mil} \and
Department of Computer Science, Faculty of Electrical Engineering,\\ Czech Technical University in Prague\\
\email{branislav.bosansky@fel.cvut.cz
}}
\maketitle  
\begin{abstract}
Reconnaissance activities precedent other attack steps in the cyber kill chain. Zero-day attacks exploit unknown vulnerabilities and give attackers the upper hand against conventional defenses.  Honeypots have been used to deceive attackers by misrepresenting the true state of the network. Existing work on cyber deception does not model zero-day attacks. In this paper, we address the question of "How to allocate honeypots over the network?" to protect its most valuable assets. To this end, we develop a two-player zero-sum game theoretic approach to study the potential reconnaissance tracks and attack paths that attackers may use. However, zero-day attacks allow attackers to avoid placed honeypots by creating new attack paths. Therefore, we introduce a sensitivity analysis to investigate the impact of different zero-day vulnerabilities on the performance of the proposed deception technique. Next, we propose several mitigating strategies to defend the network against zero-day attacks based on this analysis. Finally, our numerical results validate our findings and illustrate the effectiveness of the proposed defense approach.

\end{abstract}
\section{Introduction}
The cyber kill chain defines seven stages of cyber attack that end with gaining control of a system/network and infiltrating its data \cite{yadav2015technical}. The first stage is the reconnaissance stage in which an adversary collects valuable information regarding the network topologies, structures, node features, and the important assets of the system. To achieve this goal, an attacker may use active sensing techniques and/or passive sensing techniques. The latter can observe traffic between servers and clients to infer information from packet length, packet timing, web flow size, and response delay~\cite{schuster2017beauty}, it is difficult to detect and is invisible to the hosts running the services and can be difficult to be detected by conventional IDS. Active probing attacks send packets to a host and analyze its response. Hence, the attacker learns the system information and vulnerabilities \cite{fu2004empirical}. The active reconnaissance is faster and identifies open and unprotected ports~\cite{bansal2012detection}. On the other side, it is more aggressive and intrusive, and hence can be detected. Also, attackers may mix between active and passive attacks during the reconnaissance stage. Game theory provides a suitable framework for modeling attacks and defense against several attacks \cite{cceker2016deception,zhu2018multi,anwar2019game}.

\textbf{Cyber deception:}
Cyber deception is an emerging proactive cyber defense technology it provides credible yet misleading information to deceive attackers. Deception techniques have been used in the physical space as a classical war technique. However, deception has recently been adopted into cyberspace for intrusion detection and as a defense mechanism \cite{wang2018cyber}. Cyber deception shares some characteristics of non-cyber deception and follows similar philosophical and psychological traits. Successful deception relies on understanding the attacker's intent, tools and techniques, and mental biases. The first step to achieve this deep level of understanding is to act proactively aiming to capture the attacker and exploit the opportunity to monitor her behavior. For that purpose, honeypots are effectively used as fake units in the system/network that deceive the attacker and allow the defender to study her attack strategy, and intent in order to design a better deception scheme. 

\textbf{Honeypots:}
Among many other techniques, honeypots are widely used for cyber deception. Honeypots are fake nodes that can mislead attackers and waste their resources. They are categorized into two levels, namely low-interaction honeypots and high-interaction honeypots. Low interaction honeypots can memic specific services and are virtualized, and hence they are easier to build and operate than high interaction honeypots. However, they can be detected by adversaries more easily \cite{mokube2007honeypots}.  

\textbf{Attack Graph:}
An attack graph (AG) is a graph-based technique for attack modeling. Attack modeling techniques model and visualize a cyber-attack on a computer network as a sequence of actions or a combination of events \cite{lallie2020review}. Attack graphs (trees) are a popular method of representing cyber-attacks. There exist no unique way to instantiate attack graphs, authors in \cite{lallie2020review} surveyed more than 75 attack graph visual syntax and around 20 attack tree configurations. A key challenge of generating attack graphs is the scalability \cite{ou2006scalable}. None of the existing works has shown that the graph generation tool can scale to the size of an enterprise network. In this work, we consider a simplified attack graph where each node represents a vulnerable host in the network (i.e., it suffers one or more vulnerabilities), and edges on the attack graph represent specific exploits that provide reachability to the attacker from one host to another. In this sense, the graph scale is in the order of the size of the original network. Although this model does not explicitly model each vulnerability in the network, however, it sufficiently illustrates the attack paths that can be exploited by adversaries which are an essential input to generating optimal honeypot allocation policy. However, attack graphs can not directly model zero-day attacks since it remain unknown to the graph generating tool. Therefore, we propose parallel  

% An attack graph is a tool used to model network security by capturing network connectivity and vulnerabilities. Attack graphs could be generated in several ways. For instance, Kamdem et al. [30] generated a vulnerability multi-graph in which node are
% vulnerable hosts, and edges represent the vulnerabilities between nodes. Authors in [31]
% used attack graphs to model the causal relationship of different vulnerabilities and proposed a probabilistic metric for network security. Stochastic games played on attack
% graphs facilitate cyber deception automation and deception policy implementation on
% networks. However, the dimension of strategy space explodes for attack graph games
% in the size of the attack graph. Moreover, partial observability regarding attacker dynamics is a struggle against developing deception policy.
% A defender can model the behavior of a partially observable attacker using one of
% two approaches. A simple but naïve approach is to model the actions of the attacker as
% exogenous noise. The second approach is to take into account the attacker's actions as
% observations that are induced given her actions. Finally, the defender can assume that
% the attacker is also monitoring the system partially and is maximizing his long-term

\textbf{Zero-day attack:}
The challenge with defending against zero-day attacks is that these attacks exploit a vulnerability that has not been disclosed. There is almost no defense against exploiting such unknown vulnerability \cite{bilge2012before}. In this work, we leverage attack graphs to model potential zero-day attacks. If the considered network suffers a zero-day vulnerability then the corresponding attack graph will have some edges and hence attack paths that are unknown to the defender. Moreover, zero-day attacks are used for carrying out targeted attacks. To the best of our knowledge, this represents a new framework to proactively defend against zero-day attacks via strategic honeypot allocation based on game theory and attack graphs.

\textbf{Contributions:}
In this paper, we propose a cyber deception technique using strategic honeypot allocation under limited deception budget. We consider a game theoretic approach to characterize the honeypot allocation policy over the network attack graph. We then evaluate the deception allocation policy under zero-day attacks by introducing several vulnerabilities to the attack graph and study the sensitivity of the different potential vulnerabilities on the attacker and defender game reward. In our analysis, the defender has no information regarding the zero-day vulnerability. We identify the most impactful vulnerability location and introduce several mitigating strategies to address the possible zero-day attack. The developed game model accounts for the network topology and different importance to each node. We summarize our main contribution below:
\begin{itemize}
\item We formulate a deception game between defender and attacker to study the effectiveness of cyber deception against lateral move attacks. The game is played on an attack graph to capture relation between node vulnerabilities, node importance, and network connectivity. We characterize a honeypot allocation policy over the attack graph to place honeypots at strategic locations.
\item We evaluate the proposed deception approach against zero-day attack under asymmetric information where the attack graph is not fully known by the defender. We conduct sensitivity analysis to identify critical locations that have major impact on the deception policy in place.
\item We present three mitigating strategies against zero-day attacks to readjust the existing honeypot allocation policy based on the conducted analysis. 
\item Finally, we present numerical results for the developed game model to show the effectiveness of cyber deception as well as the zero-day attack mitigating strategies.
\end{itemize}

The rest of the paper is organized as follows. We discuss related work in section \ref{sec:relwork}. In section \ref{sec:model}, we present the system model, define the game model, and propose our deception approach.In section, \ref{sec:mitigation} we present zero-day attack mitigating strategies.
Our numerical results is presented in section \ref{sec:results}. Finally, we conclude our work and discuss future work in section \ref{sec:conc}.

\section{Related work}\label{sec:relwork}
Our research builds upon existing work on cyber deception and games on attack graphs to model zero-day attacks and characterize game-theoretic mitigation strategies.
\subsection{Cyber Deception GT:}
Game theoretic defensive deception~\cite{mu2021survey} has been widely discussed in cybersecurity research. Authors in \cite{schlenker2020game} presented a deception game for a defender who chooses a deception in response to the attacker's observation, while the attacker is unaware or aware of the deception. Authors in \cite{pawlick2015deception,cceker2016deception} proposed a signaling game based model to develop a honeypot defense system against DoS attacks. Hypergame theory~\cite{Fraser84} has been used as an extensive game model to model different subjective views between players under uncertainty. \cite{Vane99} explored hypergames for decision-making in adversarial settings. Authors in \cite{ferguson2019SHTSS},\cite{cho2019modeling} leveraged hypergames to quantify how a defensive deception signal can manipulate an attacker's beliefs.   
% and \cite{wan2021foureye} discussed hypergame-based deception against advanced persistent threat attacks performing multiple attacks performed in the stages of cyber kill chain.

\subsection{Games on AG:}
Game Theory (GT) provides a suitable framework to study security problems including cyber deception \cite{nguyen2013analyzing}. Modeling the attacker behavior allows the network admin to better analyze and understand the possible interactions that may take place over cyberspace. Security games are defender-attacker games, the defender allocates a limited set of resources over a set of targets. On other hand, the attacker goal is to attack and gain control over these targets \cite{sinha2018stackelberg}. Resource allocation problems are usually modeled as Stackelberg game where the defender leads the course of play. We consider two-player zero-sum games acting simultaneously, hence the Nash equilibrium of the game coincides with that of the Stackelberg game. Moreover, most of the resource allocation problems considered had no underlying network structure. For the cyber deception problem considered we play a security game on an attack graph which imposes a structure on the players reward function and defines the action space for both players as will be discussed in Section \ref{sec:model}.   
\subsection{Zero-day}

A zero-day attack is a cyber attack exploiting a vulnerability that has not been disclosed publicly \cite{bilge2012before}. Due to the challenges associated with zero-day attacks, authors in \cite{bilge2012before} conducted systematic study to learn the characteristics of zero-day attacks from the data collected from real hosts and identify executable files that are linked to exploits of known vulnerabilities. 

Eder-Neuhauser et al. \cite{eder2017cyber} introduced three novel classes of malware that are suitable for smart grid attacks. Their model provides a basis for the detection of malware communication to predict malware types from existing data.They suggest proper update and segmentation policies for anomaly detection. However, such an approach does not capture the dynamics of zero-day attacks or model its usage in lateral movement attacks.

Al-Rushdan et al. \cite{al2019zero} propose zero-day attack detection and prevention mechanism in Software-Defined Networks, prevent zero-days attack based on traffic monitoring. However, in practical zero-day attack incidents alter traffic information to bypass detection systems and traffic monitoring tools. In this work, we take a first step in modeling zero-day attacks in a strategic approach using game theory leveraging the existing work on cyber deception on AGs. We conduct analysis to identify impact of different vulnerabilities and propose zero-day mitigating strategies for several practical scenarios.

\section{System Model}\label{sec:model}

We consider a network represented as a graph $G_1(\mathcal{N},\mathcal{E})$ denote the network graph with a set of nodes $\mathcal{N}$ and edges $\mathcal{E}$. Nodes are connected to each other via edges modeling the reachability and the network connectivity. The defender categorizes the nodes differently according to the node's vitality and functionality. Specifically, there are two distinguishable subsets of nodes, set of entry nodes $E$ and the set of target nodes $T$, other nodes are intermediate nodes that an attacker needs to compromise along the way from the entry node (attack start node $\in E$) toward a target node $\in T$. 

\begin{figure}[htbp]
    \centering
    \includegraphics[width=10cm, height=5cm]{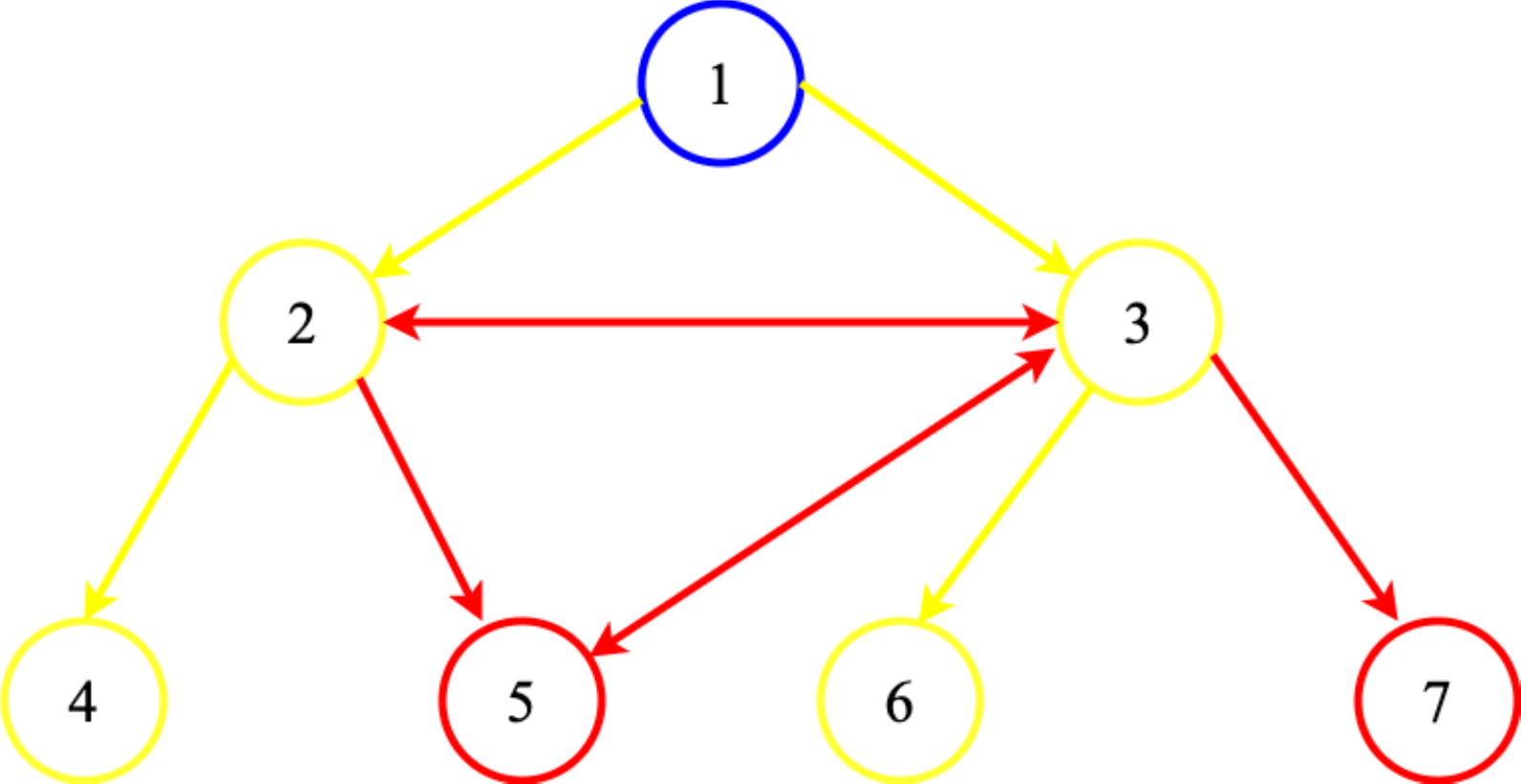}
    \caption{7-node tree network topology with single point of entry and two target nodes (5,7) and zero-day vulnerability(2,3) and (3,5).}
    \label{fig:system}
\end{figure}

An edge connecting node $u$ and node $v$, represents an exploitable vulnerability that allows an adversary to launch an attack to reach node $v$ from a compromised node $u$. In this setting, we adopt a slightly different version of the attack graph introduced in \cite{ammann2002scalable}. In other words, in this graph, each node represents a host that suffers one or more vulnerabilities that could be exploited to reach a neighboring node. Hence, the edge models the connecting link that could be used by malicious users to reach the next victim node. A legitimate user at node $v$, will have the right credentials to reach node $u$. However, an adversary will only reach $u$ through an exploitable vulnerability. For each node $i \in \mathcal{N}$ we assign a value $v(i)$. Hence, $G_1(\mathcal{N},\mathcal{E}_1)$ is an attack graph assumed to be known to the defender and the attacker.

A zero-day attack vulnerability is exclusively known to the attacker. The effect of a single zero-day vulnerability is an additional edge. This generates a different attack graph perceived by the attacker solely. Let $G_2(\mathcal{N},\mathcal{E}_2)$ denote the attack graph induced via zero-day vulnerability $e$, such that, $\mathcal{E}_2 = \mathcal{E}_1+e$.
The attacker plays a game on the graph with an additional edge(s) representing the zero-day vulnerability. In other words, considering a single zero-day vulnerability at a time, $G_2 = G_1 + \{e\}$, where $\{e\}$ is the new edge due to zero-day vulnerability.

Fig. ~\ref{fig:system} denotes 7-node tree attack graph consists of one entry node(1), 4 intermediate nodes(2,3,4,6) and 2 target nodes(5,6). In this network, one available path for reaching every target nodes. However, with two zero-day vulnerabilities(2,3) and (3,5) increases the available attack path to every target node form attacker perspective.

\subsection{Defender model}
The defender allocates one or more honeypots along the network edges as fake vulnerabilities to capture malicious traffic and illegitimate users. Let $H$ denote the honeypot budget. The defender's action is to allocate up to $H$ honeypots over the different edges. Therefore, the defender action space $\mathcal{A}_d = \left\{ \mathbf{e} \in 2^{\mathcal E}  \mid \mathbf{e}^T1 \leq H   \right\}$. Where, $\mathbf{e}$ is a binary vector of length $|\mathcal E|$, such that the $i^{th}$ entry $\mathbf{e}(i) = 1$ indicates a honeypot is allocated along the $i^{th}$ edge, and is set to zero otherwise. The defender incurs a cost associated with each action that takes into account the number of allocated honeypots. Otherwise, the defender will always max out the number of allocated honeypots. Let $C_d$ denote the cost per honeypot. Hence, the total cost is $C_d \times  \left\| a_d\right\|_1$, where $ \left\| a_d\right\|_1$ is the number of honeypots allocated by $a_d$. The defender tries to reduce the attacker reward via placing more honeypots at the edges of high potential that are attractive to attacker, while reducing the total deception cost.  

\subsection{Attacker model}
The attacker is assumed to launch a targeted attack. Therefore, she selects one of the possible attack paths to reach a target node to maximize his/her expected reward. Hence, the attacker's action space, $\mathcal A_a$, is the set of all the possible attack paths starting at an entry node $u \in E$ to a target node $v \in T$. The attacker pays an attack cost that depends on the selected attack path. We consider a cost due to traversing a node in the attack graph denoted by $C_a$. The attacker faces a tradeoff between traversing through important nodes while reducing his overall attack cost. 

\subsection{Reward function}
We define the reward function to capture the the tradeoff that faces each player. For each action profile played by the defender and attacker $(a_d, a_a) \in \mathcal{A}_d \times \mathcal A_a$, the defender receives a reward $R_d(a_d, a_a)$ and the attacker reward is $R_a(a_d, a_a)$. The defender is interested in protecting specific nodes than others. Recall that each node $i \in \mathcal{N}$ is assigned a value $v(i)$ that reflect its importance for the attacker, the defender gains more by protecting high valued nodes. On the other hand, the attacker reward increases when attacking nodes of high values along the selected attack path.
 
The defender reward is expressed as:

% \begin{equation}
%   R_d(a_d, a_a) = \sum_{u_0 \in E} \sum_{t \in T} R_{(u_0,t)}
% \end{equation}

\begin{align} \label{eq:1}
      R_d(a_d, a_a) = \sum_{i \in a_a} Cap \cdot v(i)\cdot \mathbf{1}_{\left\{ i \in a_d\right\}} \nonumber  -Esc \cdot v(i)\cdot \mathbf{1}_{\left\{ i \notin a_d\right\}}\nonumber & \\ - C_d \cdot \left\| a_d\right\|_1 + C_a (a_a)
\end{align}

where $Cap$ denotes the capture reward received by the defender when the attacker hits a honeypot along the selected attack path $a_a$. 
$Esc$ is the attacker gain upon successful attack from one node to another in the way toward the target node.  
Finally, $C_d$ and $C_a(a_a)$ are the cost per honeypot, and attack cost, respectively. The attack cost is proportional to the length of the attack path as the attacker could become less stealthy due to numerous moves. We consider a zero-sum game where $R_a + R_d =0$.
Now we readily define a two-player zerosum game $\Gamma (\mathcal{P}, \mathcal{A}, \mathcal{R})$, where $\mathcal{P}$ is the set of the two players (i.e., defender and attacker). The game action space $\mathcal{A} = \mathcal{A}_d \times \mathcal{A}_a$ as defined above, and the reward function $\mathcal{R}=(R_d, R_a)$.

Due to the combinatorial nature of the action spaces in terms of the network size, characterizing a Nash equilibrium (NE) in pure strategy is challenging. However, the finite game developed above, holds a NE in mixed strategies.
Let $\mathbf{x}_{1} $ and $\mathbf{y}_{1} $ denote the mixed strategies of defender, and attacker, when the game is played on known-to-all graph, $G_1$.
The defender expected reward of game 1 (i.e., game on $G_1$ with no zero-day vulnerabilities) is expressed as:
\begin{equation}
    U_d(G_1)=\mathbf{x}_1^TR_d(G_1)\mathbf{y}_1
\end{equation}
where $R_d(G_1)$ is the matrix of the game played on $G_1$ and the attacker expected reward $U_a(G_1) = - U_d(G_1)$. Both defender and attacker can obtain their NE mixed strategies $\mathbf{x}_1^*$ and $\mathbf{y}_1^*$ via a linear program (LP) as follows,
\begin{equation}
\begin{aligned}
&\underset{\mathbf{x}}{\text{maximize}}
& U_d \\
& \text{subject to}
& \sum_{a_d \in \mathcal{A}_d} R_d(a_d,a_a){x}_{a_d} \geq U_d, & \;\; ~~~~ \forall a_a \in \mathcal{A}_a.\\
& & \sum_{a_d \in \mathcal{A}_d} {x}_{a_d} =1,~~ & {x}_{a_d} \geq 0,
\end{aligned}
\label{eq:global_LP}
\end{equation}

\noindent
where ${x}_{a_d}$ is the probability of taking action $a_d \in \mathcal{A}_d$.  

Similarly, the attacker's mixed strategy can be obtained through a minimizer LP under $\mathbf{y}$ of $U_d$.

\section{Zero-day Vulnerability Analysis}

We conduct a zero-day vulnerability analysis by modifying the original graph $G_1$ via  considering one vulnerability at a time. The goal of this analysis is to identify the most critical zero-day vulnerability in terms of the impact of each vulnerability to the attacker's reward against a base deception strategy. The base deception strategy is $\mathbf{x}_1$ that is obtained from the game played on $G_1$. In other words, the attacker expected reward is 
$U_a(G_1) = \mathbf{x}_1^TR_a(G_1)\mathbf{y}_1$, for any game 1 mixed strategies $\mathbf{x}_1 $ and $\mathbf{y}_1$, and $U_a(G_2) = \mathbf{x}_2^TR_a(G_2)\mathbf{y}_2$ for the game played on $G_2$. The game played under $G_1$ is referred to as game 1, and the game played on $G_2$ is referred to as game 2. Where, $\mathbf{x}_{2}$ and $\mathbf{y}_{2} $ denote the mixed strategies of the game played on the $G_2$ graph (i.e., under zero-day vulnerability). 

However, $x_2$ is infeasible in practice for the defender since the defender has no information about the zero-day vulnerability nor its location. However, the attacker's action space expands to contain additional attack paths induced by zero-day vulnerabilities. Each of these vulnerabilities may produce one or more new attack path leading to the target node. 

%  Hence, both players should play the game on the same graph. The defender mixed strategy will be restricted to the same strategy gained from $G_1$. The attacker expected reward for the game played on $G_2$ is expressed as:

% \begin{equation*}
%     U_a(G_2) = \mathbf{x}_2^TR_a(G_2)\mathbf{y}_2
% \end{equation*}

Although the defender does not actually know that the network suffers a zero-day vulnerability at a specific location, he may have the knowledge that such vulnerability exists. Therefore, the attacker is not fully certain that this specific vulnerability is unknown to the defender. The attacker uncertainty regarding the defender knowledge leads to two possible game settings and hence two evaluation criteria as follows:

\begin{itemize}
    
    \item The first criterion considers an attacker that uncertain whether his opponent knows about the zero-day vulnerability. In fact, the defender has no such information, yet the attacker accounts for some infeasible defender actions. We refer to this criterion as \textit{'optimistic'}. Hence, the expected game value for the attacker,  $U^{opt}_a(G_2) = \mathbf{\hat x}_1^{T}R_a(G_2)\mathbf{y}_{2}$, where $\mathbf{\hat x}_{1}$ is a mixed strategy adopted from $\mathbf{ x}^{1}$ and padded with zeros to ensure proper matrix multiplication while it zero-enforce infeasible actions for the defender.   
    \item Secondly, we consider a \textit{'pessimistic'} criterion, where the attacker is certain that the defender does not know the zero-day vulnerability. Hence, the defender action space is exactly similar to game 1. The expected reward is  $U^{pes}_a(G_2)  = \mathbf{x}_1^{T}R_a(G_2)\mathbf{y}_{2}$. This pessimistic criterion is referred to as game 3. 
\end{itemize}

\begin{remark}
Considering one zero-day vulnerability at a time allows the defender to study the impact of each vulnerability separately, reduces the game complexity, and enables parallel analysis by decoupling the dependencies between different vulnerabilities.
\end{remark}

As explained above, we augment zeros to $\mathbf{ x}_{1}$ to restrict the defender honeypot allocation strategy and make the defender strategy consistent with the  $R_a(G_2)$ matrix. \\
Let, $\mathbf{ x}^{1} = \left [x_0, x_1, x_2, \cdots,   x_r \right ]$, and 
$\mathbf{\hat x}_{1} = \left [x_0, x_1, x_2, \cdots, x_r,\cdots, x_n\right]$. Hence, $\mathbf{\hat x}_{1} = \left [\mathbf{ x}^{1}, \cdots , x_n\right]$, where $n \geq r$, and value of all strategies from $x_{r+1}$ to $x_n$ will be zero after sorting the corresponding actions in $\mathbf{\hat x}_{1}$. For the pessimistic case (i.e., game 3), the defender is forced to play the base deception strategy $\mathbf{x}_1$ in which he also deviates from the NE of game 3. 

% The main difference between our proposed optimistic and pessimistic criteria is whether attacker is sure or not about defender is perfectly uninformed about zero-day vulnerabilities. Therefore, defender strategies for both criteria are different.

We solve one game corresponding to each zero-day vulnerability. Assume we have a set of possible zero-day vulnerabilities $\mathcal{E}_0$ such that $G_2(e)  = G_1 + e ~~~;~~ \forall e \in \mathcal{E}_0 $. For each game we record the expected attack reward, hence we sort the vulnerabilities to find the most impactful that results in the highest increase of the attacker's reward.

Without loss of generality, assuming the new vulnerability introduced one new pure strategy $a_{e}$ for the attacker (if $e \in \mathcal{E}_0$ induces more than one new attack path, we select $a_{e}$ to be the path that has higher reward), then we can establish the following theorem.

\begin{theorem}\label{theorem1}
 For the game $\Gamma$ defined in Section \ref{sec:model}, given any base policy of the defender $\mathbf{x}$, for the new attacker pure strategy $a_{e}$:
 \begin{item}
    \item $\mathbf{y}_2\left[a_{e}\right] = 1~~$;  if $U_a(\mathbf{x},s_e) > U_a(\mathbf{x},\neg a_e)  $
    \item $\mathbf{y}_2\left[a_{e}\right] = 0~~$;  if $U_a(\mathbf{x},a_e) < U_a(\mathbf{x}, a_a) ~~ \forall a_a \in \supp\{ \mathbf{y}_1 \}$ and $U_a(\mathbf{x},a_e) < U_a(\mathbf{x}, \mathbf{y}_1 )$.
 \end{item}
\end{theorem}

\begin{proof}
The proof follows strong dominance definition \cite{basar1999dynamic}.
\end{proof}
In Theorem \ref{theorem1}, we characterize two main conditions: (1) when the zero-day vulnerability generates a new attack path that strongly dominates every other existing attack path and (2) when it is being dominated by every path in terms of both pure and mixed strategy. 
%\textcolor{blue}{Difference between our evaluation schemes.}

\section{Zero-day mitigating strategies}\label{sec:mitigation}

 The defender takes additional actions to mitigate the possible zero-day vulnerability exploits. The performed game-theoretic analysis identifies the impact of each vulnerability, and the attacker's strategy for exploiting such vulnerability. 
The defender does not know which of the vulnerabilities will take place. However, to mitigate the zero-day attack, the defender allocates an additional honeypot. We consider four different strategies such as impact-based, capture-based, worst case mitigation, and critical point analysis to select the location of the new mitigating honeypot.

\subsection{Impact-based mitigation (Alpha-mitigation)}

First, we allocate based on the impact of each zero-day vulnerability. The impact measures the increase of the attacker reward due to each introduced vulnerability, $e \in \mathcal{E}_0$, where $\mathcal{E}_0$ is the set of zero-day vulnerabilities.
We allocate the new honeypot to combat the most impactful vulnerability such that, $U_a(G_2(e))$ is the highest. The defender may allocate more honeypots following the same order of impact of each $e \in  \mathcal{E}_0$. In this mitigating strategy, we assume no information is available to the defender about which vulnerability is introduced. In the next subsection, we consider the probability of each of these vulnerabilities. 
In Section \ref{sec:results}, we shed more light on the formation of the set of zero-day vulnerabilities $\mathcal{E}_0$ overcoming the possible explosion in its carnality and applying several rules to exclude dominated elements that are obviously useless or infeasible to the attacker.

\subsection{Capture-based mitigation (LP-mitigation)}

In the previous strategy (i.e., Alpha-mitigation), the defender does not account for the probability that a zero-day vulnerability may occur.
Let $P(e)$ denote the probability that a vulnerability located at edge $e \in \mathcal{E}_0 $ exists. The impact of such vulnerability is denoted by $i(e)$, where the impact is the innovation in reward received by the attacker due to exploiting $e$ on $G_1 + \{e\}$ compared to the attacker's reward on $G_1$.
Let $J(x)$ denote the cost function for the defender as follows:

\begin{equation}
    J(x) = \sum_{e \in \mathcal{E}_0} P(e) \cdot i(e) \cdot (1 - y(e)\cdot x(e) )
\end{equation}

where $y(e)$ is the probability that $e \in \mathcal{E}_0$ is exploited during an attack, and $x(e)$ is the unknown probability to assign honeypot at $e$.

The goal of the defender is to find $x \in \left[0,1\right]^{|\mathcal{E}_0|}$ that minimizes the cost function $J(x)$. This results in a linear program that can be solved efficiently. The outcome of this LP will pick the location $e$, of the highest impact and most likely to occur (i.e., $argmax_{e} P(e)\cdot i(e)$). However, since the defender may not know the probability of existence priorly, we consider a worst-case scenario. In other words, we assume that nature will play against the defender and try to minimize its reward. Hence, the defender mitigating strategy should be characterized in response to the selection of the nature that can be obtained by solving a max-min problem as discussed next.

\textbf{Worst-case mitigation (play against Nature):}
After identifying the most impactful vulnerability location or set of vulnerabilities by doing graph analysis defender does 
not sure about which zero-day vulnerabilities the attacker is going to exploit. Therefore, we do game formulation to find defender mitigating 
strategies based on the available information of the high impactful locations.

The attacker chooses one vulnerability at a time to exploit and selects a possible attack path associated with that vulnerability. Hence, the attacker's action space, $\mathcal A_n$, is the set of all possible zero-day vulnerabilities, and defender action space, $\mathcal A_{md}$, is the set of all possible high impacted locations. This problem is formulated as an auxiliary game played between the defender and nature.

For each action profile played by both players $(a_{md}, a_n) \in \mathcal{A}_{md} \times \mathcal A_n$, the defender receives a reward $R_{md}(a_{md}, a_n)$ after zero-day attack mitigation and the attacker reward is $R_n(a_{md}, a_n)$. When the attacker selects vulnerability and the defender selected high impact location, does not match the defender's mitigating reward simply comes from the defender expected reward based on which criteria we are following. When they match we just follow Eqn.(\ref{eq:1}) based on which attack path and honeypot allocation are selected by the attacker and defender respectively.

The defender reward is expressed as 

\[ R_{d}^{mitigate}(G_1) = \begin{cases} 
      U_{d}(G_2) \: or \: U_{d}(G_3) & i \neq j \\
      R_{d}^{mitigate}(a_{d}^{mitigate}, a_n) & i = j 
   \end{cases}
\]

We consider a zero-sum game. Let $\mathbf{x}_{mitigating} $ and $\mathbf{y}_{nature} $ denote the mitigating strategies of defender, and nature mixed strategies, when the game is played on known-to-all graph, $G_1$. The defender expected reward in worst case mitigation play against nature is expressed as:
\begin{equation*}
    U_{d}^{mitigate}(G_1)=\mathbf{x}_{mitigating}^TR_{d}^{mitigate}(G_1)\mathbf{y}_{nature}
\end{equation*}
where $R_{d}^{mitigate}(G_1)$ is the matrix of the game played on $G_1$ and the nature(attacker) expected reward $U_n(G_1) = - U_{d}^{mitigate}(G_1)$.

This game has played over high impact locations in graph and gives defender mitigating strategies which location to focus for mitigation and nature mixed strategies what location attacker may choose to exploit. After having defender mitigating strategies and nature mixed strategies of the attacker, any mitigation approaches can run and eventually evaluate over it.

\subsection{Critical point mitigation}

Previously, we specified a honeypot to combat zero-day attacks in addition to the honeypots used via the defender's base deception policy. Now aim at modifying the base deception policy itself to combat the zero-day attack given the outcome of our analysis of the impact associated with each zero-day vulnerability. 
An insightful observation is that the defender tends to greedily deploy honeypots in locations that are closer to the target nodes in the network. However, when the attacker chooses a different path to attack and reach the target node, this approach does not help. It is worth noting that, in defending against zero-day attacks, when the defender selects a location that protects high-degree nodes (this is captured in node values $v(i)$) that belong to multiple attack paths while being far from target nodes, the defender successfully captures the attacker more often. Interestingly, high-impact locations align with high-degree locations to be protected more often following Nash equilibrium deception strategies. 

These observations led us to conduct critical point mitigation to find overlapping locations in the graph. In critical point mitigation, we choose one of the high impacted locations which is also an overlapping location where we deploy mitigation. After having critical points we increase the cost of accessing these points, consequently re-run the game on increased cost locations to find updated defender strategies and align optimistic and pessimistic defender strategies based on the updated defender strategies.

In our critical point mitigation, we modify the base policy of the defender,  $\mathbf{x}_1$, as we do not deploy additional honeypots. After identifying the most impactful vulnerability locations, we increase node values of the nodes most affected (i.e., neighbors) by such vulnerabilities. This shifts the defender's attention to these locations and in turn, results in a modified base policy (which we refer to as critical point mitigation) that considers the significance of these nodes. We show that critical point mitigation increased the capture rate of attackers when tested in different settings such as increased number of honeypots, different vulnerable entry nodes, and target nodes.

%(1) Implementation issue as zero-day(Capture attacker by deploying ).
% Ahmed: To place an easier vulnerability to exploit at that location. Also, We can block any traffic directly from (6) and (7), Otherwise it would raise a security alert.   
% Sayed: Can we change (add new edge, that attractive to attacker) to initial graph G1 and deploy honeypot to that edge.
% We are going to change honeypot allocation from defender.

%(2) I can use only edges of G1 to mitigate(Proactive).
% Idea(Sayed): Increase number of honeypot(input) deployment.

% Idea: Increase the values  of the critical nodes and rerun the game on G1(6,7, 80,120,150).
% Idea2: Best to respond to the attacker strategy on paths from game2.

%(3)  intensive monitoring strategies.(To be mentioned in the discussion).
%(4) honeypots as entire nodes!  ( Maybe attach them to node (7)).

%(5) How about vulnerability types (OS, Windows, etc). One type will generate k edges at a time. (Extension)
% Chain of vulnerabilities. (Done)
% Routing as defense against attack paths. (To be mentioned in the discussion).

\section{Numerical Results}\label{sec:results}

In this section, we present our numerical results to validate the proposed game-theoretic model. We evaluate our analysis of zero-day vulnerability and the proposed mitigating strategies. First, we present game results to identify impact of possible zero-day vulnerabilities for the optimistic and pessimistic defender. Second, we show the results of the proposed deception and mitigation strategies. Finally, we discuss our findings, limitations, and future directions of our current research.

%\textcolor{red}{ - talk about the network, number nodes, action space for the defender, action space for attacker. for game 1 section A. \\
%- Then present the results for game 1. put some figure for attacker reward game 1 and game 2\&3 for top 8 edges. in subsection B. \\
%- Then New subsection  for mitigation. Put figure for the no-mitigation, alpha, LP, random.}

\subsection{Experiment:}
The initial honeypot allocation strategy follows the Nash equilibrium of the game model (game 1). We formulate the problem as a zero-sum game, solve the game defined in Section 4 and find the Nash equilibrium in terms of the mixed strategies, $x^{*}$ and $y^{*}$, for the defender and attacker strategies, respectively.

For the analysis of the potential impact of zero-day vulnerabilities, we consider a 20-node network topology with 22 edges shown in Fig. ~\ref{fig:network}.  
The 20-node network topology shown in Fig. ~\ref{fig:network} represents an
attack graph with multiple root node (i.e., the entry node $\mathcal{E}$ = \{0, 1, 2\}). In this scenario we define the set of target nodes as three nodes, T = \{18, 19, 20\}.

% Initially, the defender selects some honeypot location in the network to prevent the attacker from reaching the target in the network. Our initial game formulation is based on that. Then we come up with two new formulations of the game based on zero-day vulnerabilities that the attacker is perfectly informed about but the defender does not.

To form the set of zero-day vulnerabilities $\mathcal{E}_0$, we analyze our 20 nodes network topology. If we consider all possible new edges in the network which is impractical as time complexity is ${n}^2$ for $n$ nodes network. Since the vulnerability analysis is independent between different elements in $\mathcal{E}_0$, it can run in parallel computing nodes to reduce the run-time. Moreover, some locations are practically infeasible or useless to the attacker. We implemented a set of rules that excluded useless edges and edges that do not benefit our analysis. For instance, we excluded edges leading to dead-end nodes and edges from nodes that are unreachable from any attack path. Also, direct edges between targets and entry nodes are dominant edges without further analysis.

\begin{figure}[htbp]
    \centering
    \includegraphics[width=13cm, height=8cm]{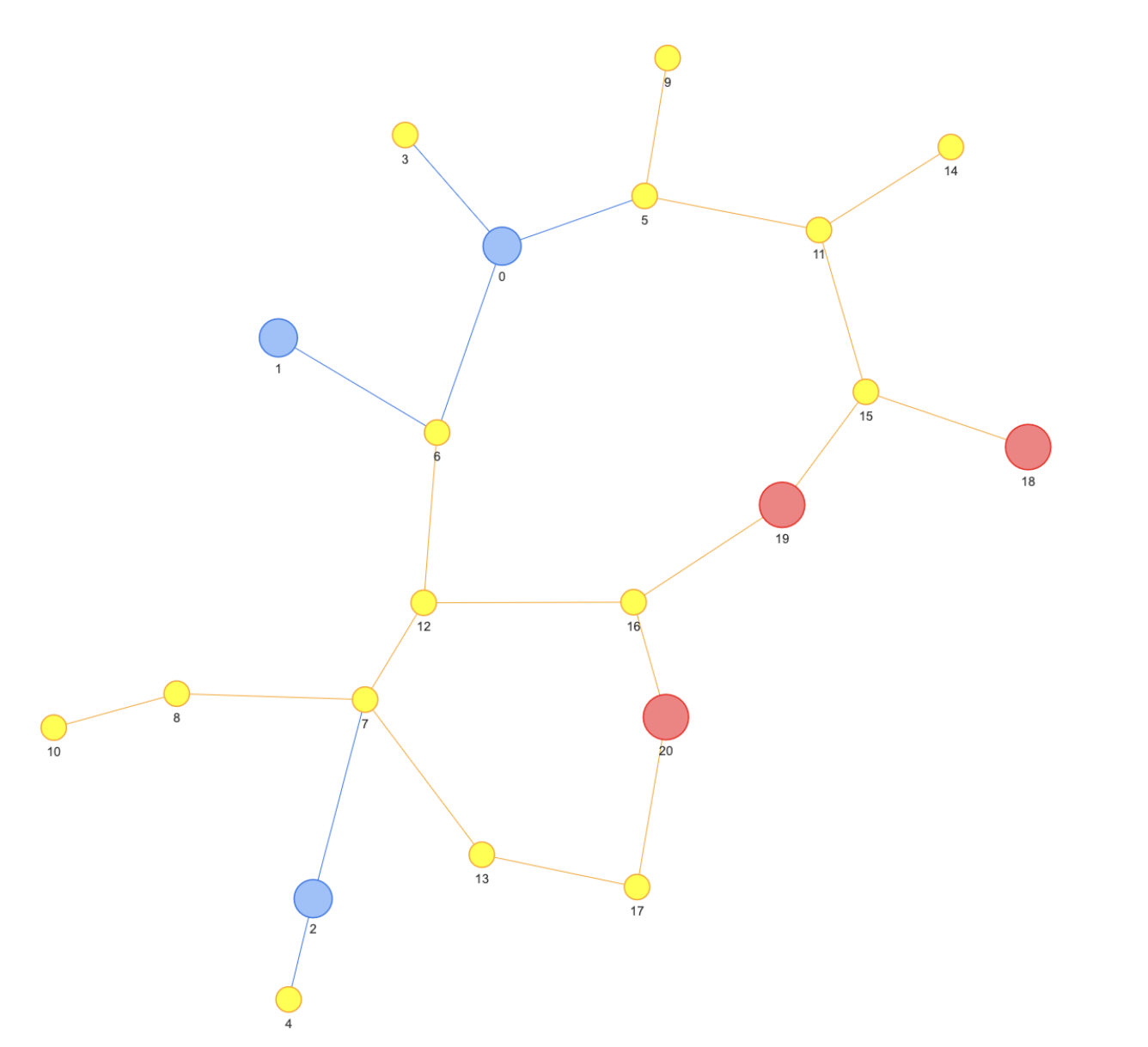}
    \caption{Network topology of 20 nodes with blue red, and yellow color for entry, target and intermediate nodes respectively}
    \label{fig:network}
\end{figure}

We compare the Nash equilibrium strategy for honeypot allocation with existing attack policies to illustrate the effectiveness of our proposed cyber deception approach. For that, we observe defender and attacker gain under different conditions. Fig. ~\ref{fig:capesc} illustrates how defender reward change on several conditions including escape reward of honeypot and capture cost of the attacker.

Fig. ~\ref{fig:esc} shows the defender reward against different cost values for escaping a single honeypot in the network over different attack policies. We also compare the Nash equilibrium reward against the greedy and random attacker. A greedy attacker always selects nodes that have the highest values to attack regardless of their cost. A random attacker does not that informed about network that is unable to distinguish between possible attack paths and hence uniformly selects among available attack paths.

When the attacker deviates from equilibrium strategies  $\mathbf{y}^{*}$, such as choosing greedy or random strategies, the defender reward tends to be higher or the same. For low Esc values, defender reward against greedy attacker higher compare to Nash equilibrium which less motivates an attacker to play rational strategies. On the other hand, high Esc values lure the attacker to take more risk in attacking valuable nodes, and as a consequence a gradual decrease in defender reward.

Fig. ~\ref{fig:cap} shows a linearly increase in defender reward for different attacker policies. For high cap values, defender reward increases if the attacker deviates from rational strategies. 
% The analysis of different escaping costs and capture reward brings the emergence of developing rational strategies as naive approaches are very sub-optimal.

\begin{figure*}[h!]%
    \centering
    \begin{subfigure}{.50\columnwidth}
        \includegraphics[width=\columnwidth, height=5cm]{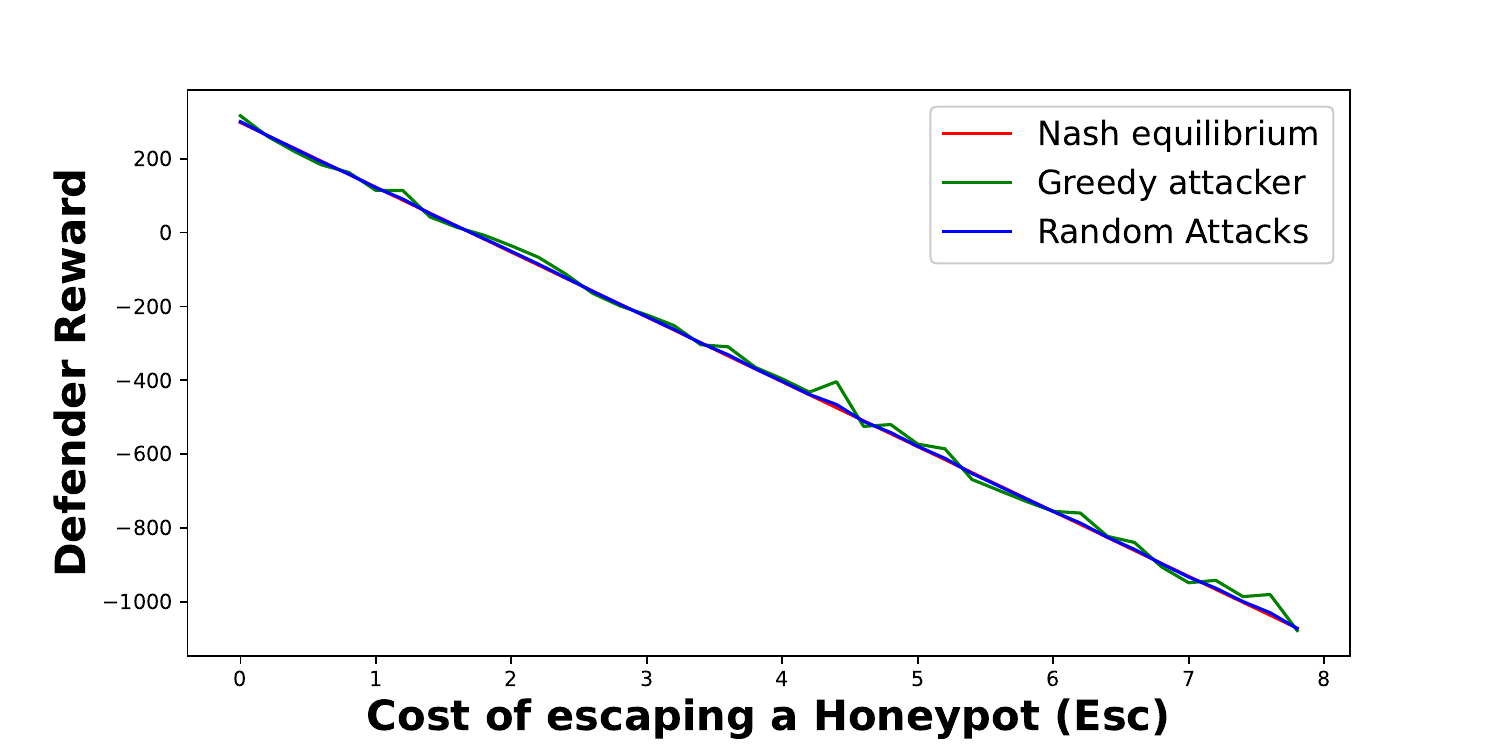}%
        \caption{Defender reward versus attacker \\ cost of escaping a honeypot.}%
        \label{fig:esc}%
    \end{subfigure}\hfill%
    \begin{subfigure}{.50\columnwidth}
        \includegraphics[width=\columnwidth, height=5cm]{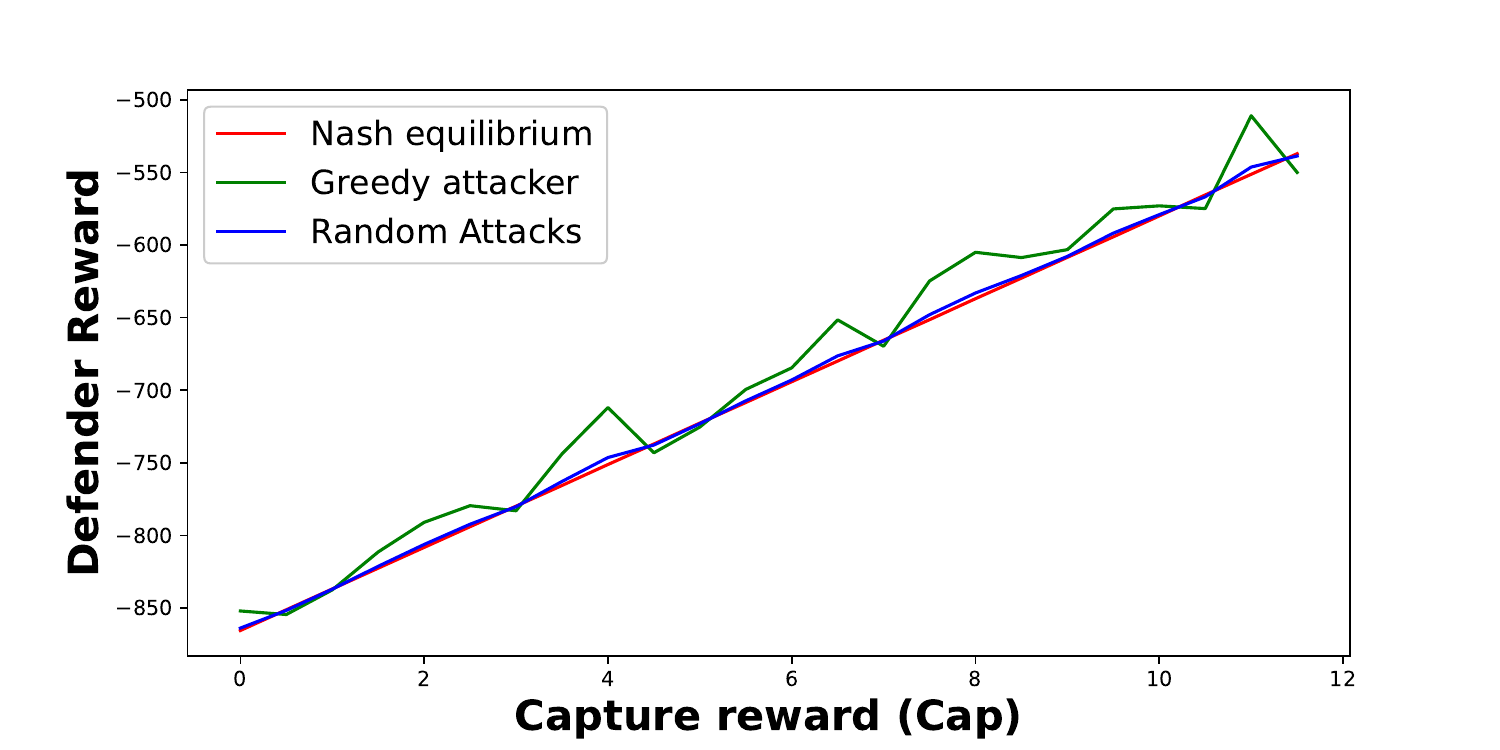}%
        \caption{Defender reward versus  capture cost \\ of attacker.}%
        \label{fig:cap}%
    \end{subfigure}\hfill%
    \caption{Defender reward over different Cap and Esc values. }%
    \label{fig:capesc}%
\end{figure*}

\comment{
\begin{figure} 
    \centering
    \includegraphics[width=13cm, height=8cm]{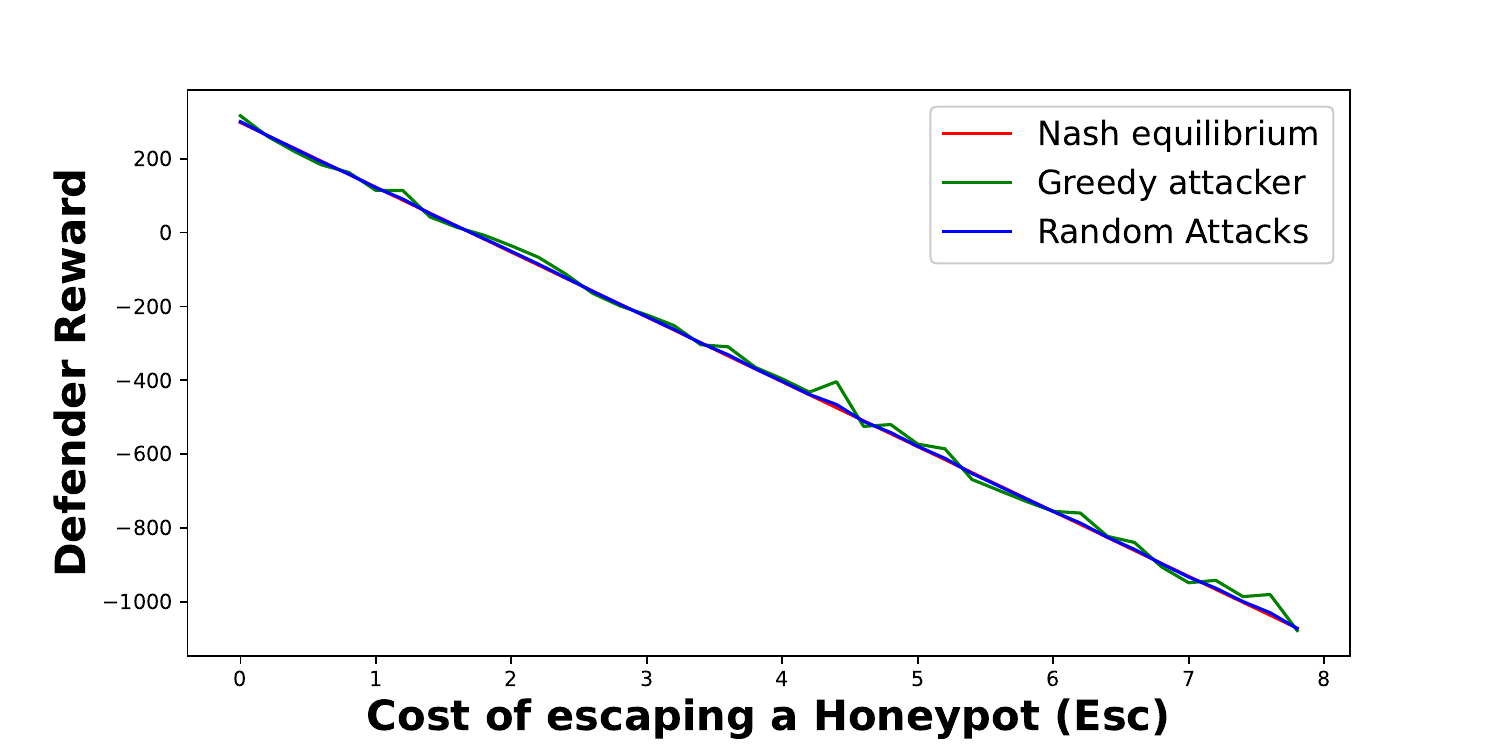}
    \caption{Defender reward versus attacker cost of escaping a honeypot.}
    \label{fig:esc}
\end{figure}

\begin{figure} 
    \centering
    \includegraphics[width=13cm, height=8cm]{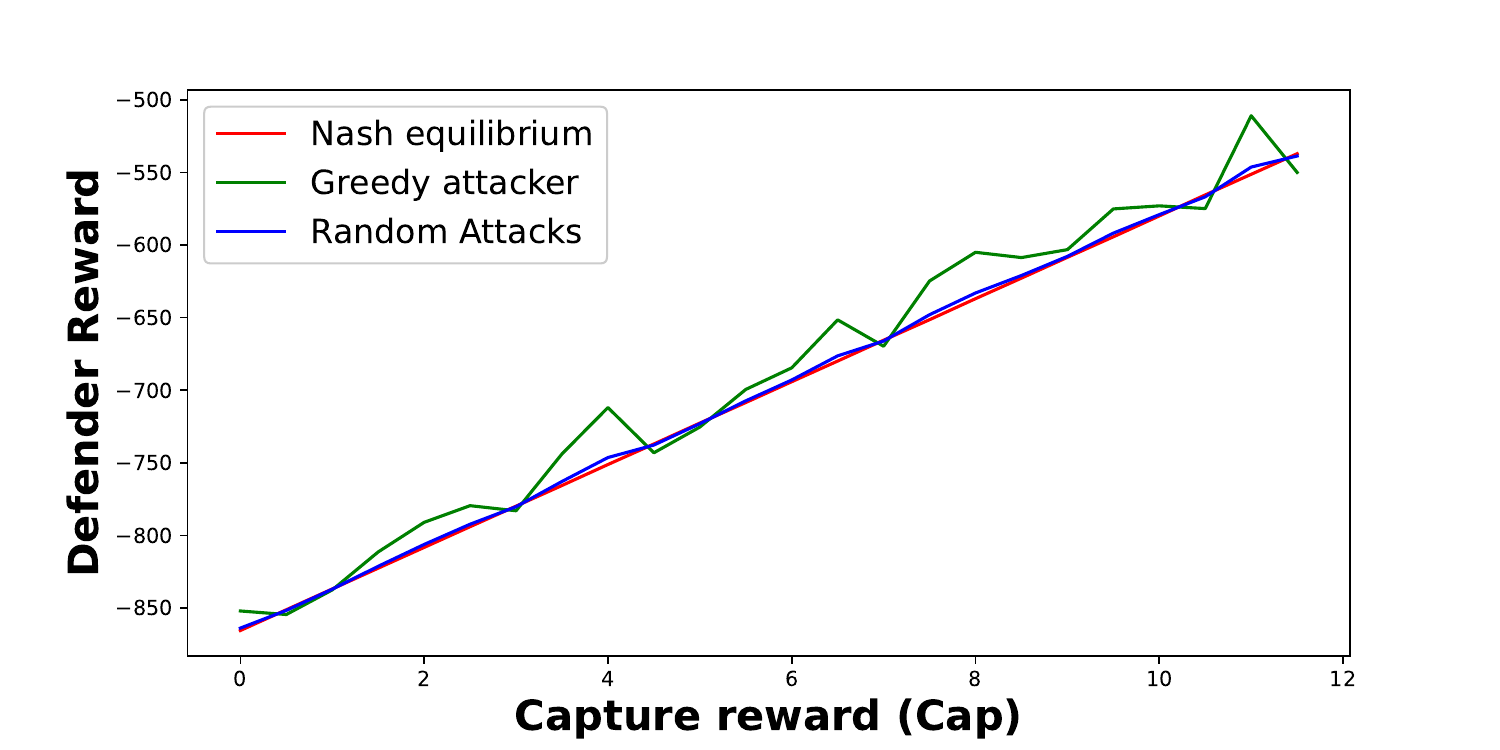}
    \caption{Defender reward versus  capture cost of attacker.}
    \label{fig:cap}
\end{figure}
}

In addition, we also examine attacker's reward against different defender policies to deceive attacker and protect the network. Fig. ~\ref{fig:nhcen} shows how attacker gain decreases as the number of honeypots increases and its dependence on the entry nodes. 

\begin{figure*}[h!]%
    \centering
    \begin{subfigure}{.50\columnwidth}
        \includegraphics[width=\columnwidth, height=5cm]{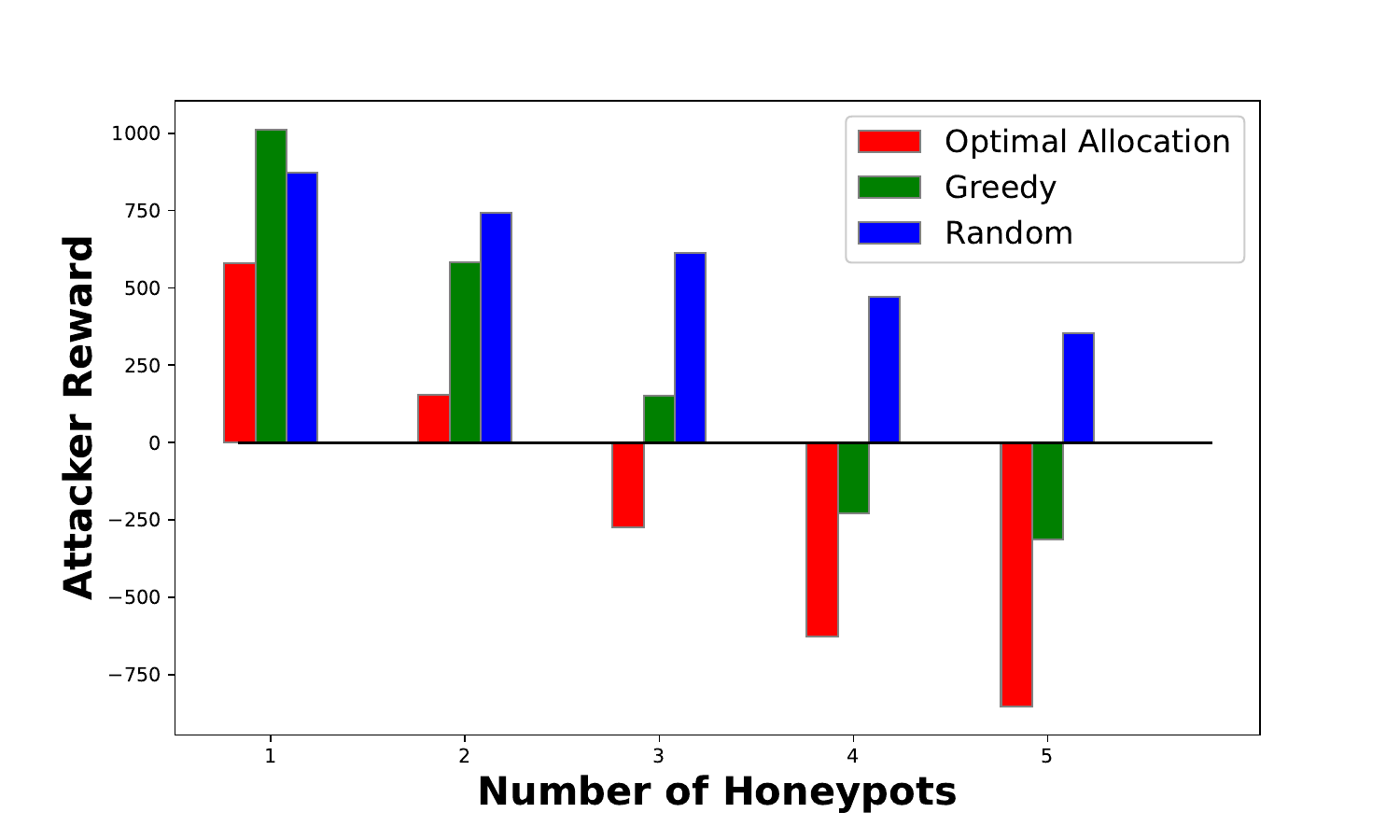}%
        \caption{Attacker reward versus the \\ number of honeypots.}%
        \label{fig:nh}%
    \end{subfigure}\hfill%
    \begin{subfigure}{.50\columnwidth}
        \includegraphics[width=\columnwidth, height=5cm]{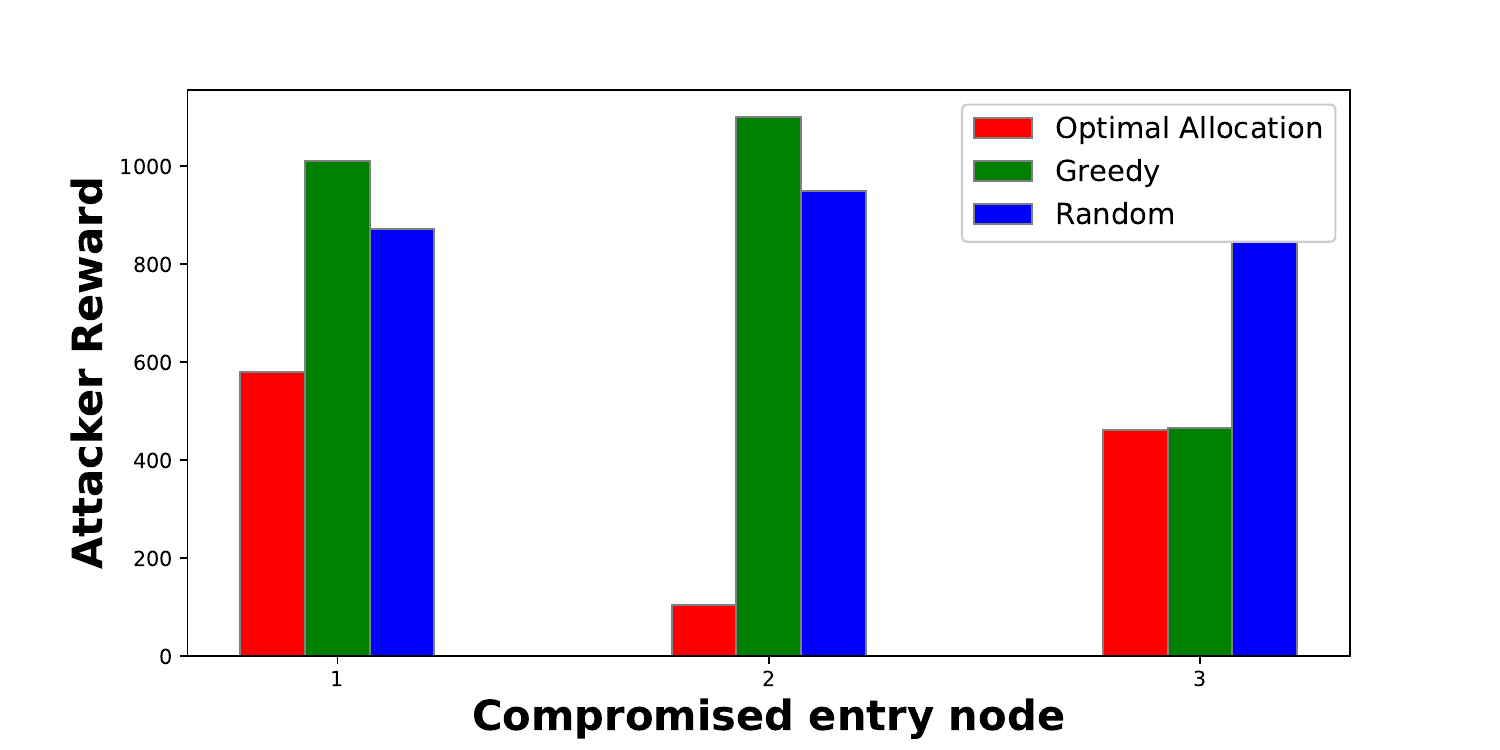}%
        \caption{Attacker reward versus compromised entry node.}%
        \label{fig:cen}%
    \end{subfigure}\hfill%
    \caption{Attacker reward over different condition such as variation in honeypot number and compromised entry node.}%
    \label{fig:nhcen}%
\end{figure*}

\comment{
\begin{figure} 
    \centering
    \includegraphics[width=13cm, height=8cm]{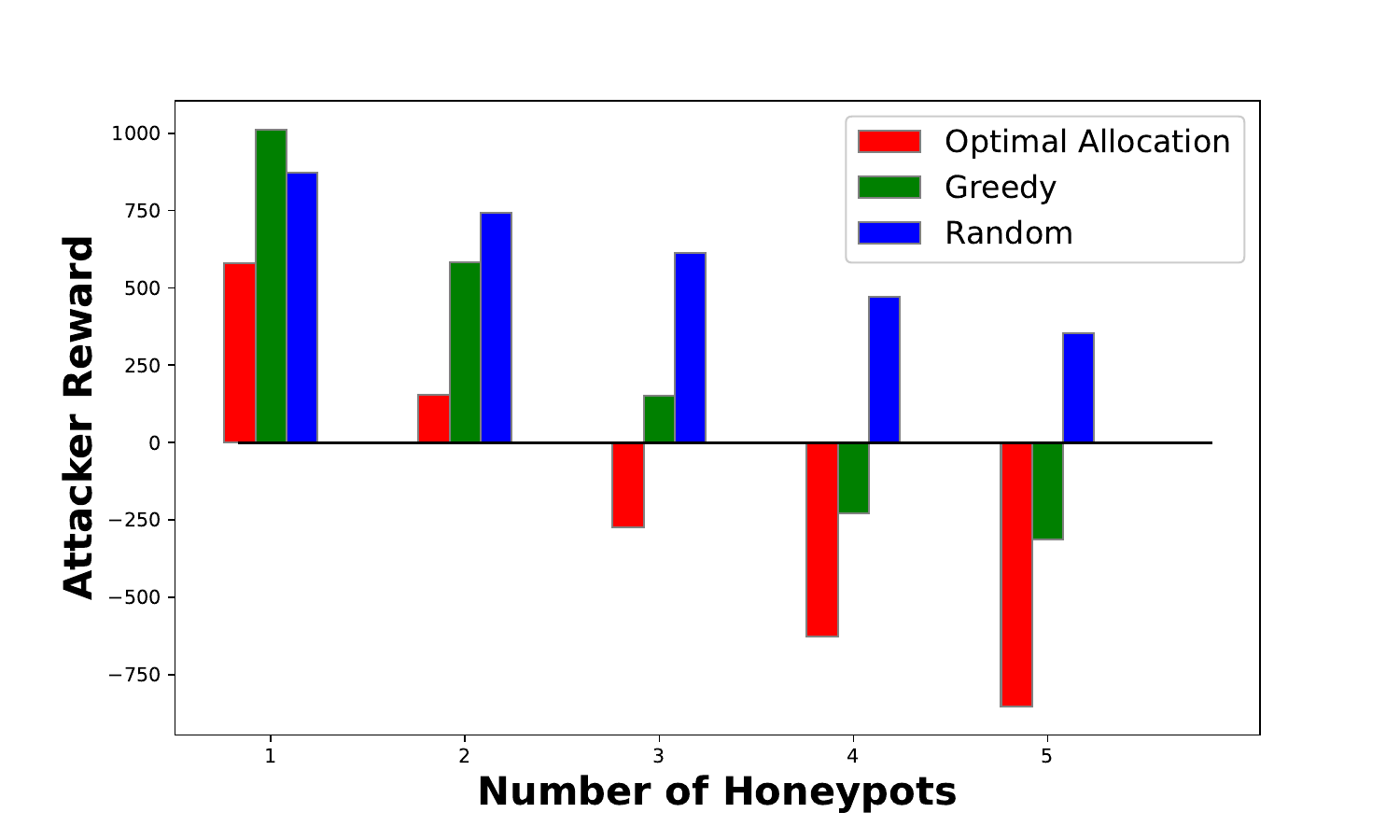}
    \caption{Attacker reward versus the number of honeypots.}
    \label{fig:nh}
\end{figure}
}

In Fig. ~\ref{fig:nh} we plot the average attack reward for different defender policies on honeypot budgets. We compare the performance of our optimal allocation with the greedy allocation that always allocates honeypots in the path of highest values nodes and random policies where defender uniformly select one node to protect rather than considering network topology analysis. The analysis of Fig. ~\ref{fig:nh} illustrates that the optimal budget of honeypot in this network is three or more honeypots, as it dramatically reduces the effect of the attack. Also deploying 3 or more honeypots is very costly.

In our 20-node network, three entry nodes are compromised at the start of the attack, so the attacker can attack using all possible existing paths in the network starting from any of the compromised entry nodes.
We also plot the attacker’s reward for a different number of compromised nodes in the network as shown in Fig. ~\ref{fig:cen} over different defender policies. Here greedy and optimal allocation produces the same magnitude result.  

Both Fig. ~\ref{fig:capesc}  and Fig. ~\ref{fig:nhcen} illustrate that deviating from Nash equilibrium and selecting some naive policies would not be optimal. Developing optimal mitigating strategies against a well-informed attacker critical for the defender to outperform naive deception policies such as random or greedy policies.

\comment{
\begin{figure}[h!]
    \centering
    \includegraphics[width=13cm, height=8cm]{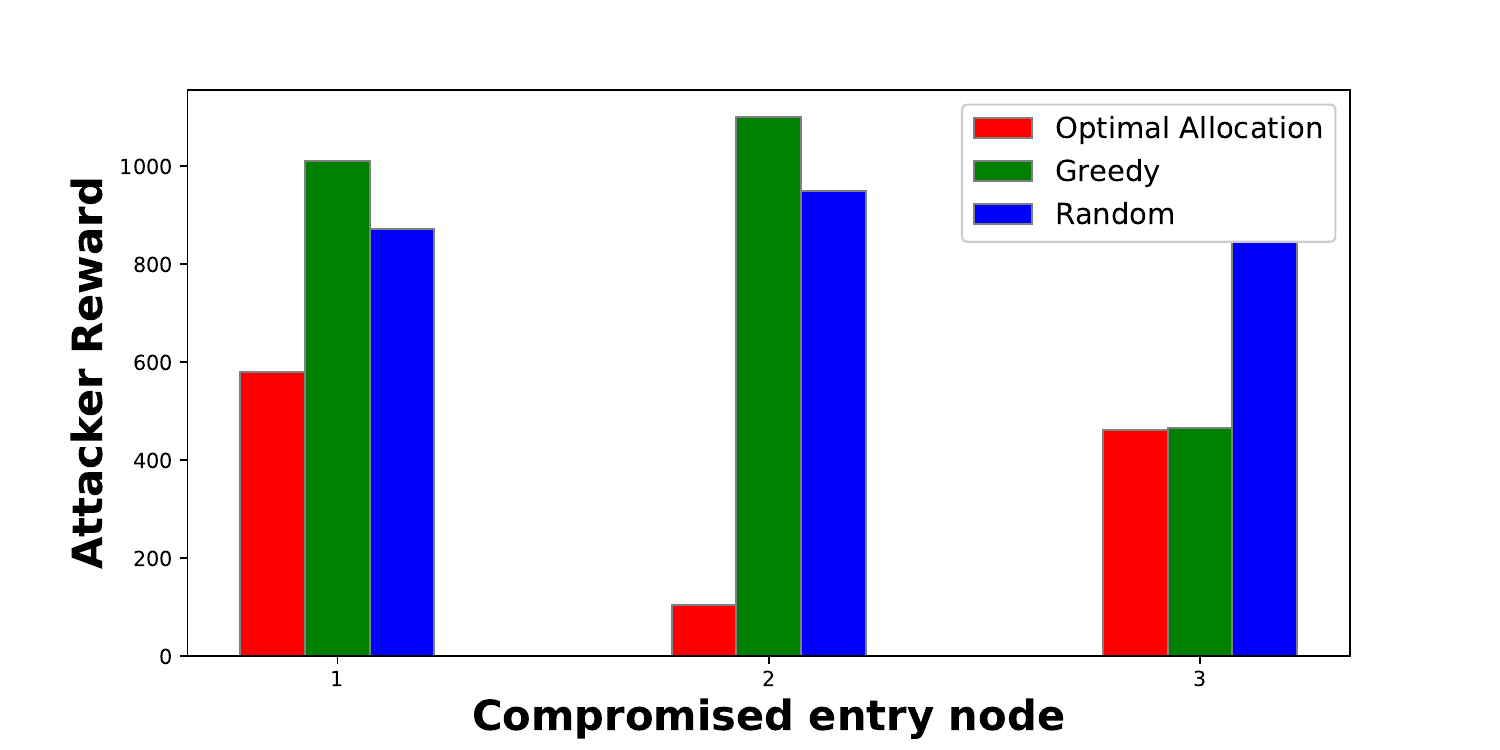}
    \caption{Attacker reward versus compromised entry node.}
    \label{fig:cen}
\end{figure}
}

\subsection{Impact of zero-days vulnerability}
In our analysis, we find out high impact locations (zero-day vulnerabilities) for the 20-node network. We measure the impact of zero-day vulnerability. We consider two scenarios, first, when the attacker is certain that the defender does not know zero-day vulnerability. Second, the attacker is not sure whether the defender is aware of these zero-day vulnerabilities. 

We also observe some zero-day vulnerabilities increase attacker reward massively and some remain the same compared to the naive defender. It is worth mentioning that some zero-day vulnerabilities also increase attacker reward in both cases, some increase only one scenario, not both depending on the reward function.

In Table. ~\ref{table:1} we present attacker reward for different high-impact locations against the naive, optimistic, and pessimistic defender. Attacker reward against naive defender is the benchmark, attacker reward against pessimistic and optimistic defender defines how impactful that zero-days is.  

\begin{table}[!htb]
    \captionsetup{justification=centering}
    \caption{Attacker reward against naive defender, optimistic defender, pessimistic defender for top 10 edges}
    \label{table:1}
    \centering
    \small
    \begin{tabular} {|c|c|c|c|} 
        \toprule
\thead {Edge}  
    & {\thead{Naive defender}} & {\thead{Optimistic defender}} & {\thead{Pessimistic defender}} \\
    \midrule
(6, 7) & 153.43 & 401.03 & 398.70 \\
(5, 7) & 153.43 & 374.65 & 370.26 \\
(3, 7) & 153.43 & 344.39 & 345.20 \\
(16, 17) & 153.43 & 326.52 & 326.52 \\
(12, 13) & 153.43 & 325.54 & 325.55 \\
(15, 17) & 153.43 & 323.60 & 323.60 \\
(11, 13) & 153.43 & 322.42 & 322.42 \\
(11, 17) & 153.43 & 315.72 & 315.71 \\
(14, 17) & 153.43 &  313.99 & 313.99\\
(12, 17) & 153.43 &  307.24 & 307.23\\
    \bottomrule
    \end{tabular}
\end{table}

\subsubsection{Attacker reward increases: } 

Based on our study, we highlight several reasons why certain zero-day vulnerabilities cause high damage to the defender compared to others. First, if a zero-day vulnerability creates multiple attack paths to any or all target nodes, that challenges the defender base-deception policy with limited honeypots in place and hence, causes significant damage. Second, zero-day vulnerabilities that are very close to any target nodes on the attack graph empower the attacker through a shortcut and enhance her reward. Also, a combination of the first two features leads to a significant loss for the defender.

% For some cases the reward of the attacker is slightly higher against the optimistic defender compared to the pessimistic defender. The reason is that the attacker gain due to the advantage that he does not sure whether the optimistic defender knows about zero-day vulnerabilities whereas he is sure that the pessimistic defender is imperfectly informed about zero-day vulnerabilities.

\subsubsection{Attacker reward remain same: } Interestingly, not all potential zero-day vulnerabilities cause significant damage to defender in terms of increasing attacker reward. Such zero-day vulnerabilities do not add useful actions to attacker action spaces that benefits the defender, consequently, the defender does not need to take mitigating measures for these types of vulnerabilities. Therefore, these observations benefit the defender to develop proactive defense focusing on most critical vulnerability locations.

\subsection{Mitigation:}

As detailed in Section 5, we proposed several approaches to develop  mitigating strategies against zero-day attacks. In our approach, the defender goal is to thwart the attacker's progress in the network by observing network information. We present numerical results to show the effectiveness of our mitigating approaches such as measuring proportion under various settings.
\begin{figure*}[h!]%
    \centering
    \begin{subfigure}{.50\columnwidth}
        \includegraphics[width=\columnwidth, height=5cm]{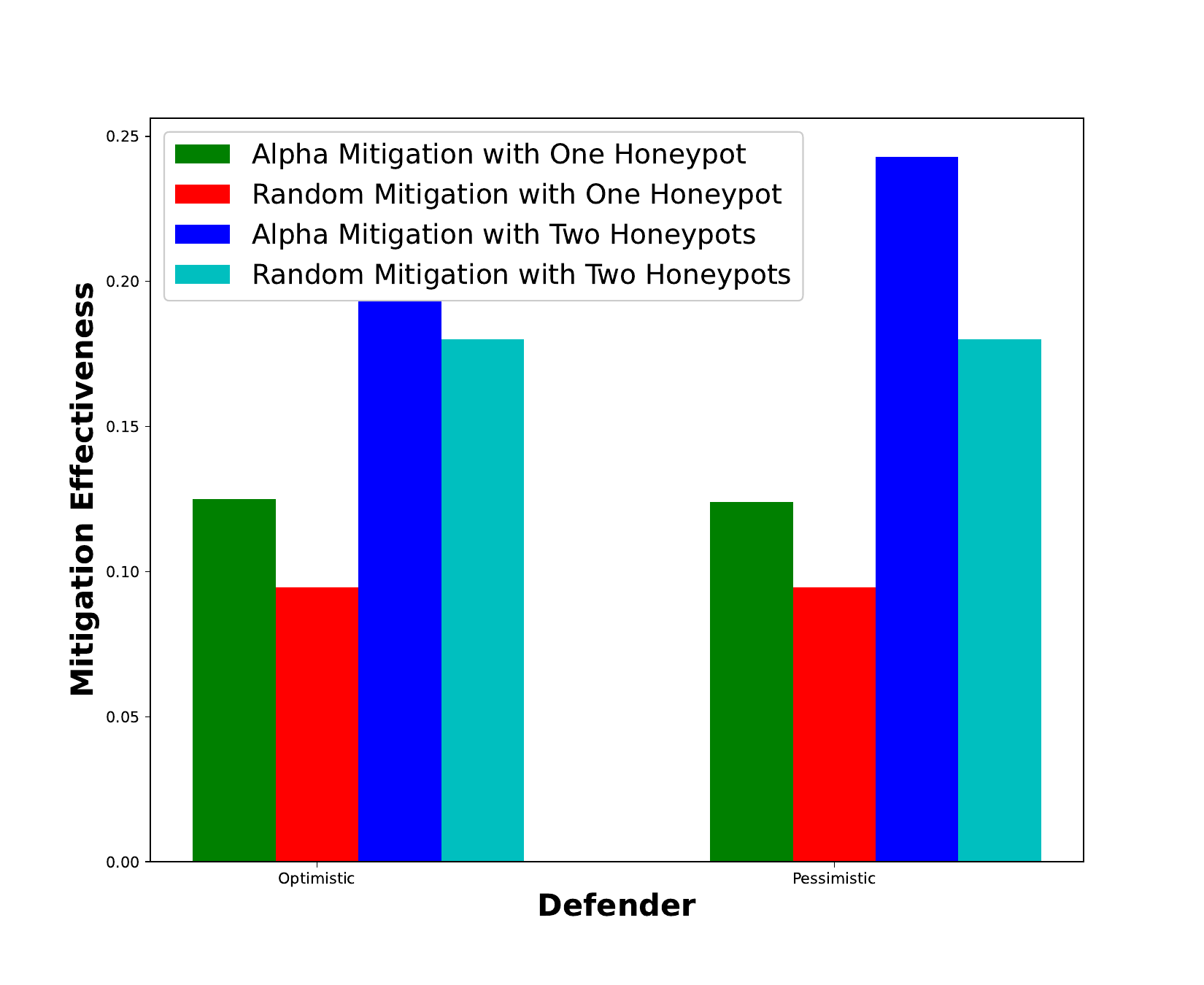}%
        \caption{Capture proportion of defender \\ over zero-days vulnerabilities with \\ single and   multiple honeypot for \\ both Alpha and random strategies.}%
        \label{fig:alpha}%
    \end{subfigure}\hfill%
    \begin{subfigure}{.50\columnwidth}
        \includegraphics[width=\columnwidth, height=5cm]{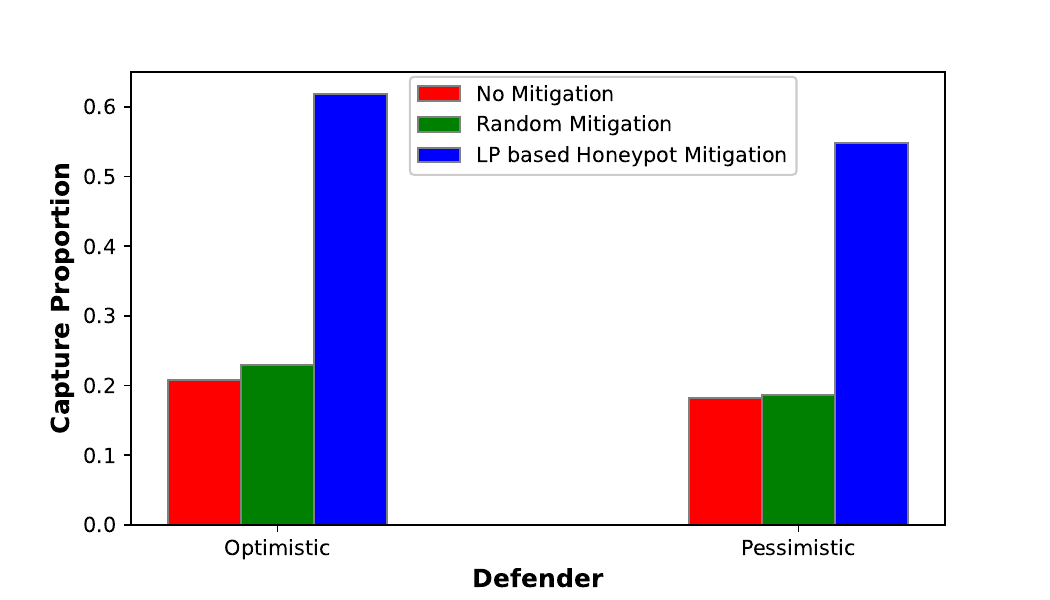}%
        \caption{Capture proportion of defender  \\ with no, random,  and LP-based \\ honeypot mitigating strategies.}%
        \label{fig:lp}%
    \end{subfigure}\hfill%
    \caption{Attacker capture proportion over different mitigating strategies \\ of defender including Alpha or LP.}%
    \label{fig:alphalp}%
\end{figure*}
% \vspace{-0.5cm}
%\subsubsection{Alpha:}

In Fig. ~\ref{fig:alphalp} we show the proportion of attacker capture both for the optimistic and pessimistic defender with impact and linear programming-based mitigation. Fig. ~\ref{fig:alpha} presents the result of our impact-based mitigation. In our Alpha mitigation, we place honeypot based on the high-impact location whereas random strategies choose a location uniformly to place honeypot. Mitigation effectiveness denotes the percentage of zero-day vulnerabilities defender mitigation (Alpha) prevents among all vulnerabilities. And capture proportion denotes the percentage of time an attacker is captured when exploiting a particular vulnerability.

Fig. ~\ref{fig:alpha} shows optimistic defender Alpha mitigation with one honeypot has higher mitigation effectiveness compared to random mitigation with one honeypot. On the other hand, the same strategies with 2 honeypots show a higher degree of deviation compared to the previous which denotes an increasing number of honeypots is useful but not a feasible solution.  

Fig. ~\ref{fig:lp} denotes attacker capture proportion over no, random, and LP-based honeypot mitigation both for the optimistic and pessimistic defender. No-mitigation and random mitigation are very close to each other meaning that randomly allocating honeypots will not bring any gain. After having the probability of allocating mitigating honeypot at different locations by solving a linear program explained in Section \ref{sec:mitigation}, we place honeypot on the corresponding location(s) and measure the proportion of capture the attacker increases.

%It is worth noting that the capture proportion of Fig. ~\ref{fig:alpha} and Fig. ~\ref{fig:lp} are not comparable as one represents the percentage of vulnerability capture from all zero-day vulnerabilities and another denotes the percentage of time attacker being captured when selects particular vulnerability.

%Fig. ~\ref{fig:honeypot2} actually compares different defender strategies for one and two mitigating honeypots. Increasing number of mitigating honeypot decreases attacker reward. 
\comment{
\begin{figure}[h!]
    \centering
    \includegraphics[width=12cm, height=7cm]{game3_2.png}
    \caption{Attacker reward for different edges over defender strategies with multiple mitigating honeypots in game3}
    \label{fig:honeypot2}
\end{figure}}

%Fig. ~\ref{fig:proportion_alpha} shows capture proportion of attacker over different defender strategies both for optimistic and pessimistic defender. Increased number of mitigating honeypots increased proportion of attacker being captured. 

\comment{
\begin{figure}[h!]
    \centering
    \includegraphics[width=12cm, height=7cm]{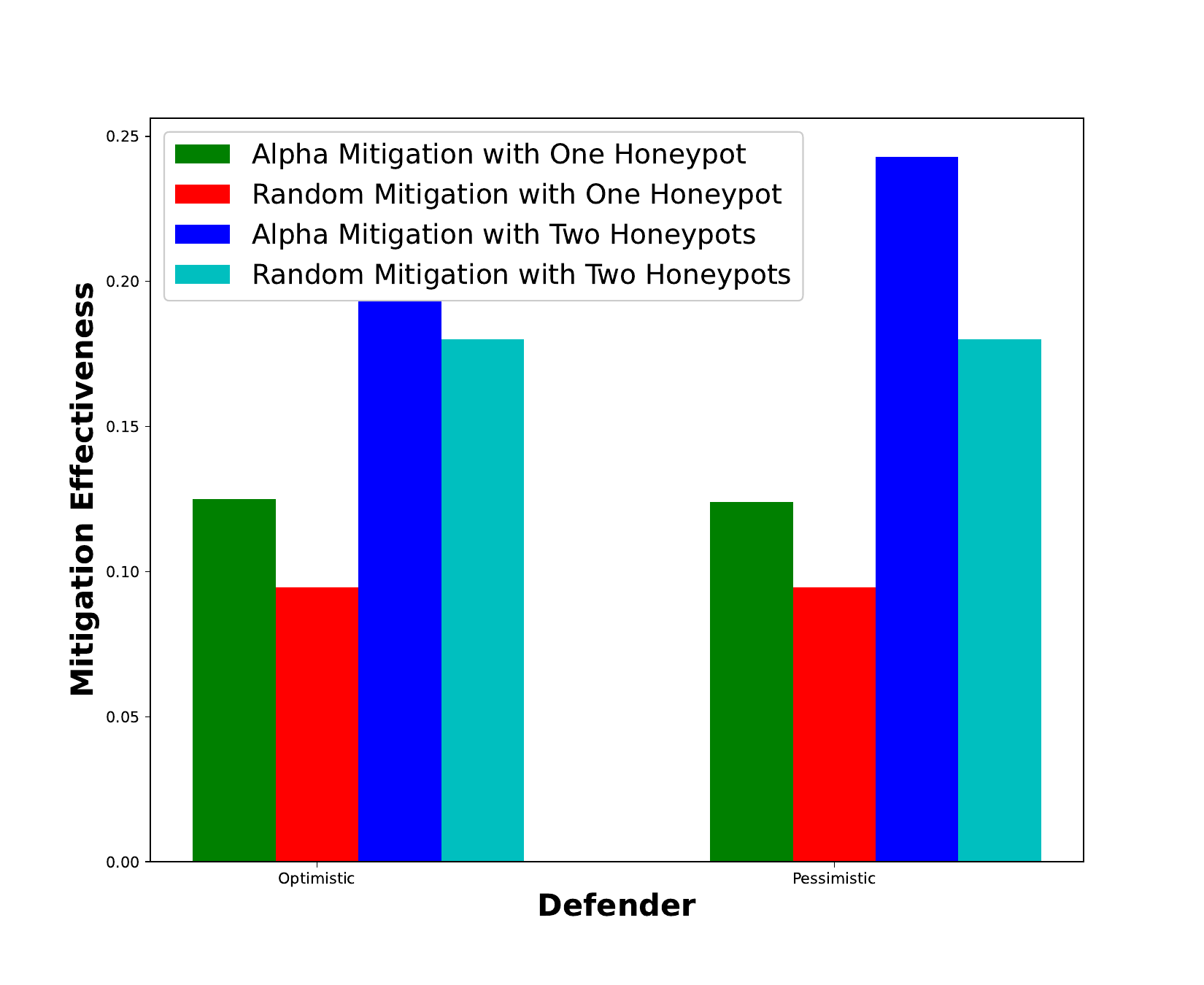}
    \caption{Capture proportion on different defender strategies}
    \label{fig:proportion_alpha}
\end{figure}}

%\subsubsection{Capture-based mitigation:}
\comment{
\begin{figure}[h!]
    \centering
    \includegraphics[width=12cm, height=7cm]{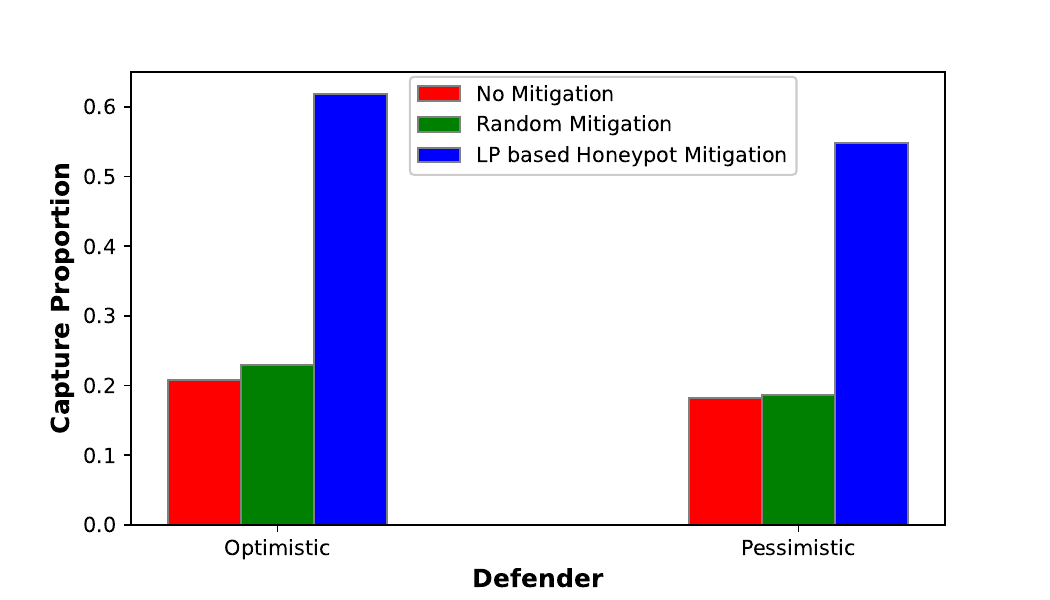}
    \caption{Capture proportion on different defender mitigation strategies}
    \label{fig:lp}
\end{figure}}
%\subsubsection{Worst-based mitigation:}

%\subsubsection{Critical point analysis:}

\begin{figure*}[h!]%
    \centering
    \begin{subfigure}{.50\columnwidth}
        \includegraphics[width=\columnwidth, height=5cm]{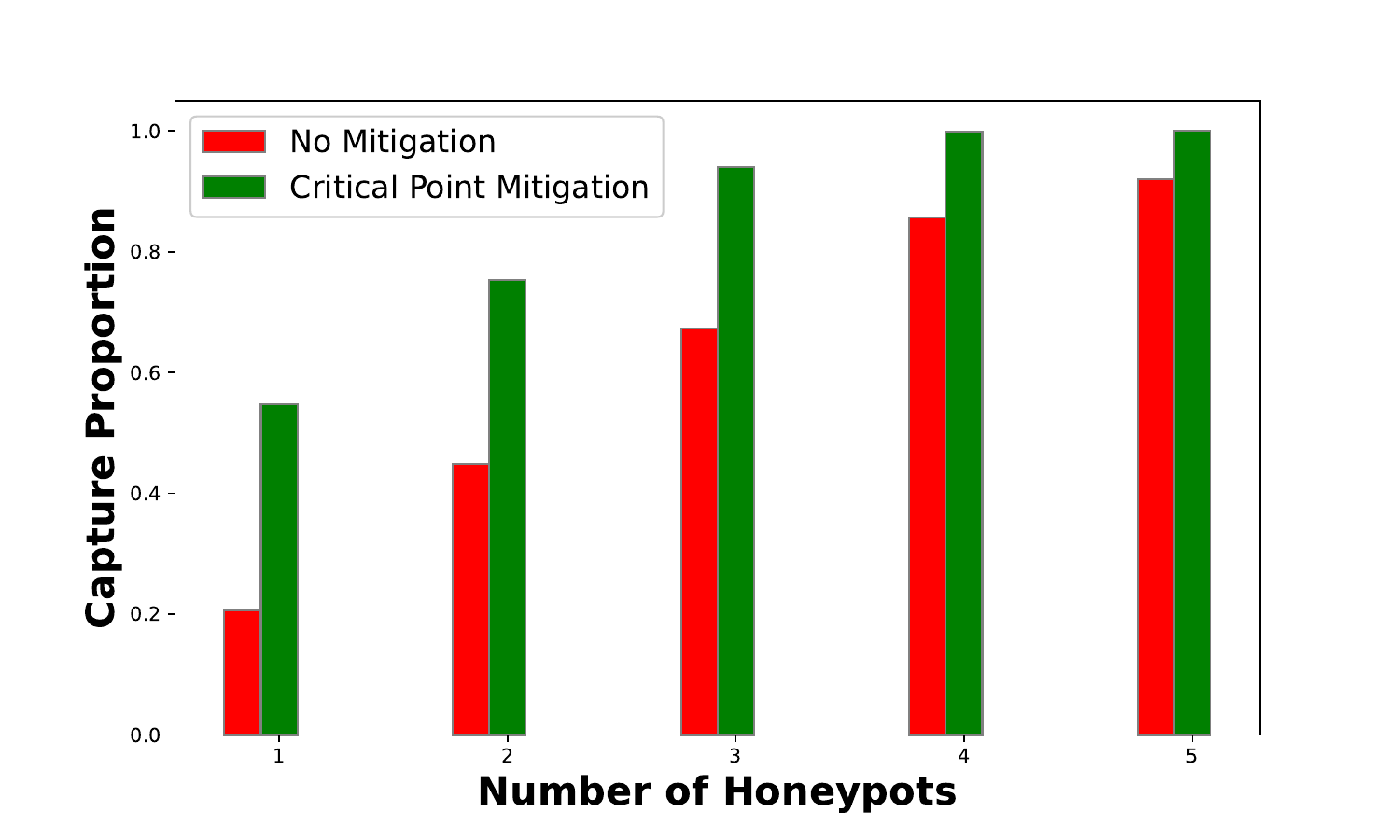}%
        \caption{Proportion of capture of attacker \\ versus the number of honeypots by \\ optimistic defender.}%
        \label{fig:cnh}%
    \end{subfigure}\hfill%
    \begin{subfigure}{.50\columnwidth}
        \includegraphics[width=\columnwidth, height=5cm]{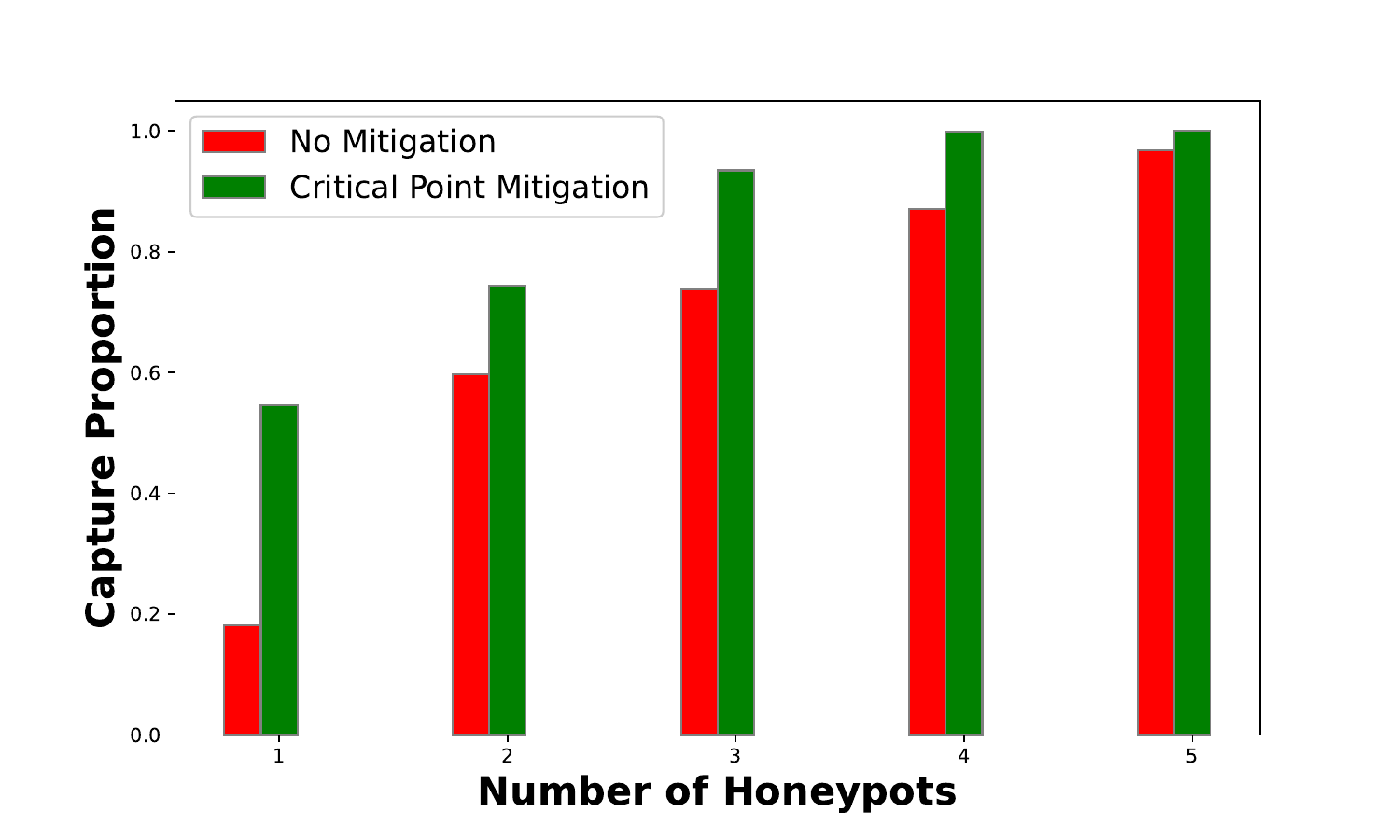}%
        \caption{Proportion of capture of attacker \\ versus the number of honeypots by \\ pessimistic defender.}%
        \label{fig:ccen}%
    \end{subfigure}\hfill%
    \caption{Attacker's capture proportion over different number of honeypots \\ for defender on before and after critical node mitigation.}%
    \label{fig:cnhcen}%
\end{figure*}
% \vspace{-0.5cm}

In critical point mitigation, we modify the base policy without additional honeypot to take into account the criticality of the most impactful vulnerability. Fig. ~\ref{fig:cnhcen} denotes the capture proportion over the increased number of honeypots for the defender.

Fig. ~\ref{fig:cnh} shows capture rate increase for both no-mitigation and critical point mitigation strategies at different deception budgets (numbers of honeypots in the base policy). The difference between no and critical point mitigation reduces over the increased number of honeypots, which denotes that the number of honeypots more than three is not useful for mitigation. Fig. ~\ref{fig:ccen} shows the same result compare to Fig. ~\ref{fig:cnh}.

\begin{figure}[h!]
    \centering
    \includegraphics[width=10cm, height=6cm]{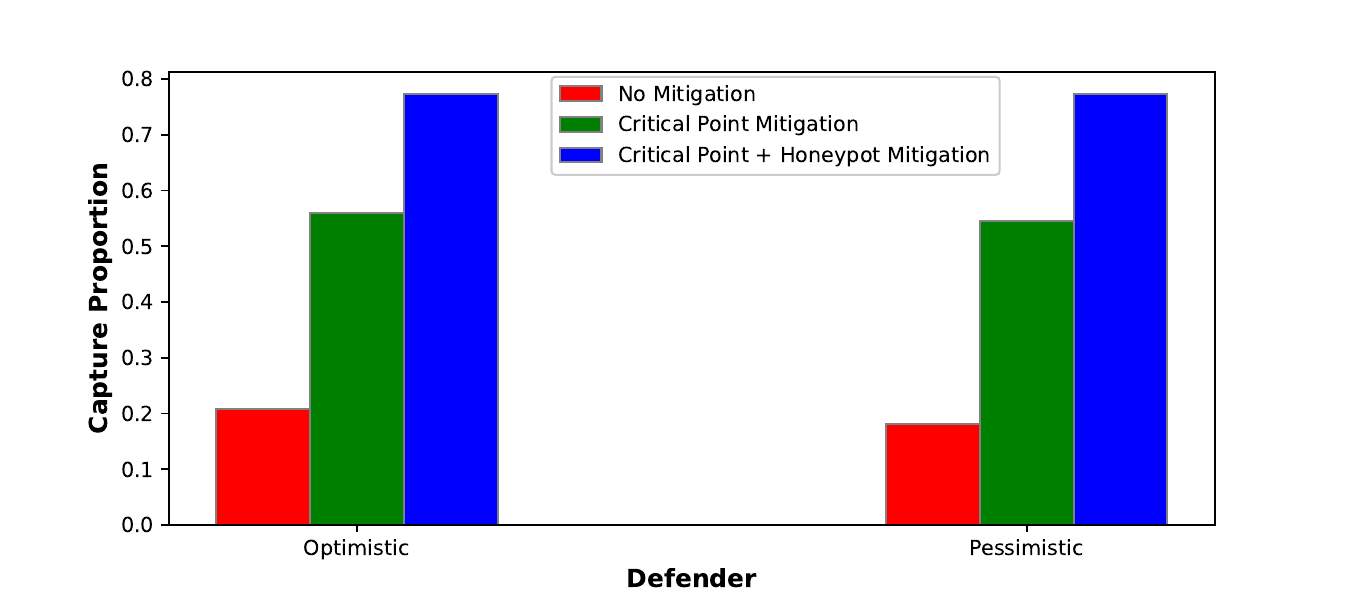}
    \caption{Capture proportion on critical point based defender mitigation strategies}
    \label{fig:cpoint}
\end{figure}

Fig. ~\ref{fig:cpoint} demonstrates capture proportion on different defender mitigation strategies. Critical point mitigation and critical point mitigation with added honeypot outperform no mitigation. It is worth noting that adding one honeypot with critical point mitigation is useful as it increases the proportion of capturing the attacker.

% TODO:
% \textcolor{blue}{Interestingly, LP-based mitigation with additional honeypot (shown in Fig. ~\ref{fig:lp}) in terms of capture rate is less than the capture rate of the critical point mitigation strategy shown in Fig. ~\ref{fig:cpoint}. This may seem counter-intuitive as mitigation with additional honeypot should have higher performance compare to the mitigation without honeypot. However, LP-based mitigation as the defender regardless of its type imperfectly informed about the location where to deploy honeypot to thwart zero-day attack which is not the case in critical point mitigation. Additionally, critical point mitigation is different than any other mitigation as it modifies the base policy of the defender, we also observe its reflection as added honeypot with critical point mitigation is not helpful.} 

% \begin{figure*}%
% \centering
% \begin{subfigure}{.49\columnwidth}
% \includegraphics[width=\columnwidth]{figure20nodes.eps}%
% \caption{A 20-node network topology.}%
% \label{fig:topology20nodes}%
% \end{subfigure}\hfill%
% \begin{subfigure}{.49\columnwidth}
% \includegraphics[width=\columnwidth]{figure30nodes.eps}%
% \caption{A 30-node network topology }%
% \label{fig:topology30nodes}%
% \end{subfigure}\hfill%
% \caption{The two generated network topologies with randomly generated edges.}%
% %\label{fig: exp3}%
% \end{figure*}

\section{Conclusion}\label{sec:conc}
% The conclusion goes here.
% https://www.overleaf.com/project/62150a05665b9f44c586713c
In this paper, we proposed a security resource allocation problem for cyber deception against reconnaissance attacks. We proposed a novel framework to assess the effectiveness of cyber deception against zero-day attacks using an attack graph. We formulated this problem as a two-player game played on an attack graph with asymmetric information assuming that part of the attack graph is unknown to the defender. We identified the critical locations that may impact the defender payoff the most if specific nodes suffer a zero-day vulnerability. The proposed analysis is limited to considering a single vulnerability at a time, and focusing on the node location. Our future work will consider a set of vulnerabilities at a time which will follow the proposed analysis while significantly increasing the action space of the game model.  

% conference papers do not normally have an appendix

% use section* for acknowledgment
% \section*{Acknowledgment}

% The authors would like to thank...

%
% ---- Bibliography ----
%
% BibTeX users should specify bibliography style 'splncs04'.
% References will then be sorted and formatted in the correct style.
%
% \bibliographystyle{splncs04}
% \bibliography{mybibliography}
% \section*{Acknowledgment}

%
% ---- Bibliography ----
%
% BibTeX users should specify bibliography style 'splncs04'.
% References will then be sorted and formatted in the correct style.
%

%\bibliographystyle{splncs04}
\bibliographystyle{unsrt}
\bibliography{reference}
\end{document}